\documentclass[pra,aps,amsmath,amssymb,groupedaddress,floatfix,nofootinbib]{revtex4}

\usepackage[utf8]{inputenc} 
\usepackage{amsmath}
\usepackage{amssymb}
\usepackage{mathrsfs}
\usepackage{bm}
\usepackage{color}
\usepackage{float}
\usepackage{amsthm}
\usepackage{appendix}
\usepackage{subfig}
\usepackage{url}

\usepackage{graphicx}
\DeclareGraphicsExtensions{.pdf,.jpg,}

\usepackage{todonotes}

\newcommand{\alice}{{\mathcal{A}}}
\newcommand{\bob}{{\mathcal{B}}}

\newcommand{\ds}{\displaystyle}

\newtheorem{theorem}{Theorem}

\begin{document}

\author{P. Yadav$^{1,2}$, P. Mateus$^{1,2}$, N. Paunkovi\'c$^{1,2}$, A. Souto$^{2,3}$}
\affiliation{$^{1}$Instituto de Telecomunica\c{c}\~oes, 1049-001 Lisboa, Portugal}
\affiliation{$^{2}$Departmento de Matem\'atica, Instituto Superior T\'ecnico, Universidade de Lisboa, Av. Rovisco Pais, 1049-001 Lisboa, Portugal}
\affiliation{$^{3}$LaSIGE, Departamento de Inform\'atica, Faculdade de Ci\^encias, Universidade de Lisboa, 1749-016 Lisboa, Portugal}


\title{Quantum contract signing with entangled pairs}
\begin{abstract}
We present a quantum scheme for signing contracts between two clients (Alice and Bob) using entangled states and the services of a third trusted party (Trent). The trusted party is only contacted for the initialization of the protocol, and possibly at the end, to verify clients' honesty and deliver signed certificates. The protocol is 
fair, i.e., the probability that a client, say Bob, can obtain a signed copy of the contract, while Alice cannot, can be made arbitrarily small, and scales as $N^{-1/2}$, where $4N$ is the total number of rounds (communications between the two clients) of the protocol. Thus, the protocol is optimistic, as cheating is not successful, and the clients rarely have to contact Trent to confirm their honesty by delivering the actual signed certificates of the contract. Unlike the previous protocol (Paunkovi\'c {et al.}, {\em Phys. Rev. A} {\bf 84}, 062331 (2011)), in the present proposal, a single client can obtain the signed contract alone, without the need for the other client's presence. When first contacting Trent, the clients do not have to agree upon a definitive contract. Moreover, even upon terminating the protocol, the clients do not reveal the actual contract to Trent. Finally, the protocol is based on the laws of physics, rather than on mathematical conjectures and the exchange of a large number of signed authenticated messages during the actual contract signing process. Therefore, it is abuse-free, as Alice and Bob cannot prove they are involved in the contract signing process.
\end{abstract}

\maketitle

\section{Introduction}
Quantum cryptography traces back to late 1960s and early 1970s work on quantum money by Stephen Wiesner. While this work was published only a decade later, in~1983~\cite{wie:83}, it had a significant impact on what usually is considered the birth of quantum cryptography, the~seminal BB84 paper on quantum key distribution (QKD)~\cite{ben:14}. Secure quantum communication offers higher, unconditional (i.e., information theoretic) security levels, as~opposed to the computational security of classical cryptography. As~a consequence, it became among the most prominent of the emerging quantum technologies, and~QKD systems are currently available on the market, such as ID Quantique, QuintessenceLabs, etc. Other related protocols, such as quantum secure direct communication~\cite{long:02}, have also been developed. Nevertheless, there is much beyond key distribution that quantum cryptography can offer, such as quantum secret sharing~\cite{hill:99}, quantum private query~\cite{giova:08,gao:19}, and~quantum secure distributed learning~\cite{sheng:17}. Secure multiparty computation (SMC)~\cite{lin:pin:09} presents another class of cryptographic protocols in which the privacy of users' data and inputs is protected. Instances~of such schemes include private data mining and e-voting, to~name a few. Recently, quantum solutions' bit commitment and oblivious transfer, cryptographic primitives that allow for execution of SMC, have~been proposed~\cite{bro:sch:16}. In~the current work, we present a quantum solution to the contract signing~problem. 

Contract signing~\cite{even:yacobi:80} is a security protocol that falls within the group of the so-called commitment protocols~\cite{loura:14,almeida:16,loura:16}. In~general, the~protocol can be defined for an arbitrary number of parties (clients engaged in the protocol). For~simplicity, we discuss the case of a two-party protocol, which can be straightforwardly generalized to an arbitrary number of~participants. 

The participants, usually referred to as Alice and Bob, have a common contract upon which they decide to commit, or~not. The~commitments are traditionally done by simple signatures: having a text of the contract with Bob's signature stamped on it, Alice can appeal to the authorities (the Judge), which in turn declares the document valid (i.e., binds the contract). In~other words, having Bob's signature gives Alice the power to enforce the terms of the contract. Consequently, signing his name on a copy of the contract means that Bob commits to the contract. The~aim of a contract signing protocol is that either both clients obtain each others' commitments or none of them do (the protocol is said to be fair). Further, if~both clients follow the protocol correctly, both of them can obtain each others' commitments with certainty (the protocol is then said to be viable).

If only Alice has a copy with Bob's signature (i.e., only Bob is committed), she can later in time choose to either enforce, or~not, the~terms of the contract. Bob, however, has no power whatsoever: his~future behavior is determined solely by Alice's decisions. For~example, Alice may have a document with Bob's signature on it, declaring that he would buy a car from her, for~a fixed amount of money. Knowing that only she has such a document, Alice can continue to negotiate the price of her car with other potential customers: in case she obtains a better offer, she is free to discard Bob's offer and thus is able to earn more money. Bob does not have such an option: if Alice does not obtain a better offer, she can always force Bob to buy the car from her, by~showing to the authorities the contract signed by Bob. Having no proven commitment (signature) from Alice, Bob cannot enforce the contract himself and is thus unable to prevent Alice from such behavior, which puts him in an unfair~situation. 

Achieving fairness is trivial when clients meet up and simultaneously sign copies of the contract, thus both obtaining each others' commitments. Unfortunately, doing so when the clients are far apart, e.g.,~over the Internet, is difficult: indeed, sending his signed copy to Alice gives Bob no guarantee that he will obtain one from Alice; on the other hand, obtaining a signature from Alice before actually sending his gives Bob the advantage of having Alice's commitment without committing~himself.

It has been shown~\cite{fis:lyn:pat:85,even:yacobi:80} that the fairness of a contract signing protocol with spatially-distant clients can be achieved only by introducing a trusted third party, usually referred to as Trent, during~the phase of exchanging clients' commitments. Trent's role is to receive clients' commitments and perform the exchange only upon obtaining signed copies of the contract from both clients. However simple and straightforward this solution may seem, it has a drawback, as~Trent (in practice, a~trusted agency accredited by the State that offers its time and resources for the exchange of money, e.g.,~public~notaries) may be expensive. Therefore, the~need for protocols using third parties as little as possible arises. Some contract signing protocols~\cite{even:83,gold:84,even:85} do not require a trusted third party, but~use a number of transmissions to send the pieces of signatures, or~the partial information required to obtain the complete signature, in~each message. Another possible way out is to design optimistic and/or probabilistic~protocols.

In optimistic contract signing protocols~\cite{asokan:97}, the exchange of commitments is, unless~something goes wrong, executed solely by Alice and Bob. Only in case communication between the clients is interrupted (malfunction of the network, etc.), a~trusted third party is involved~\cite{asokan:98}. In~probabilistic protocols~\cite{ben:90,rabin:83}, by~exchanging messages between each other, clients increase their probabilities to bind the contract. To~be (probabilistically) fair, such protocols have to ensure that at each stage of the information exchange, the~probabilities to bind the contract of both clients can be made arbitrarily close to each other (no client is significantly privileged). One such protocol is~\cite{rabin:83}, for~which the symmetry between the clients' positions is strengthened by the requirement that the joint probability that one client can bind the contract, while the other cannot, can be made smaller than any given $\varepsilon < 1$. Finally, there is both a probabilistically fair and optimistic solution, with~an optimal number of exchanged messages~\cite{ben:90} for which even a stronger fairness condition is satisfied: the conditional probability that a client cannot bind the contract, when the other has already done so, can be made arbitrarily~low.

Recently, a~probabilistically fair and optimistic quantum protocol was presented in~\cite{paun:2011} (see a version using the simultaneous dense coding scheme in~\cite{situ:15}). There, the~trusted third party, Trent, was required to initiate the protocol and was contacted later only in case something went wrong. The~protocol in~\cite{paun:2011} was also abuse-free~\cite{garay:99}, i.e.,~the clients cannot provide proofs of being involved in a contract signing procedure. Nevertheless, it has three important disadvantages: (i) Alice and Bob have to share the content of the contract with Trent; (ii) both clients have to be present in order to bind the contract, in~case something goes wrong and Trent's services are required; and (iii) they have to agree upon the content of the contract before the protocol initialization. In~this paper, we propose an improved version of the contract signing protocol where (i) the clients never disclose the content of the contract to Trent, (ii) only one client is needed to bind the contract, and~(iii) the clients can decide upon the contract after they initially contact~Trent.

Regarding Point (iii) from the previous paragraph, note that often, when parties initiate business negotiations, this does not result in making a deal formalized by a contract. Thus, involving Trent, who charges for his services, might often result in the waste of clients' resources, and~Point (iii) might seem not to present a real advantage. Nevertheless, waiting for the last moment and contacting Trent only upon successful agreement and contract formulation might result in the system's failure due to possible communication bottlenecks. Imagine the following situation. Alice and Bob negotiate buying/selling a certain product, say a security system, knowing that on a given date in the future, a~big company will announce a new model with its novel performances. Obviously, the~price of the model used will highly depend on the information about the new one, and~Alice and Bob will only upon learning the new piece of information decide upon the final contract. The~problem is, many other users may choose to make similar business contracts in the same period of time, and~if they all have to only then contact Trent, this might cause a communication bottleneck and the failure of the system. Thus, being able to contact Trent in advance and then, only later, ``offline'' (without contacting the trusted agency) exchange the commitments might be useful, especially in ``more serious'' business deals including larger amounts of~money.

In Section~\ref{sec:setup}, we begin with the description of the contract signing protocol with all the different phases explained in detail. In~Section~\ref{sec:security}, we provide the security analysis of the protocol, together~with relevant numerical results. Finally, in~Section~\ref{sec:conclusion}, we present the conclusions and discuss the contributions to the~area.
 
\section{The proposed protocol}\label{sec:setup}

In the quantum contract signing proposal~\cite{paun:2011}, Trent sends two strings of qubits, one to Alice and another to Bob, such that each qubit is randomly prepared in one of the four BB84 states. The~commitments are done by measuring one of the two observables on all qubits given to a client: in case of accepting the contract, measurement in the computational basis is performed, while choosing to reject it, one measures in the diagonal basis. Since the two bases are mutually unbiased, as~a consequence of the Heisenberg uncertainty principle, it is impossible to learn both properties simultaneously. Thus, measurement outcomes of each client serve as certificates of the commitments. In~order to achieve the fairness criterion (the two commitment choices have to be the same), the~measurements are done in rounds, such that only one qubit per client is measured in each round and the outcomes are exchanged. Since, in~addition to qubit strings, Trent also informs Alice of Bob's states (and vice~versa), exchanging the outcomes allows clients to check each others' honesty. In~our current proposal, instead of accompanying qubits with the classical information about other client's states, Trent exchanges entangled pairs. As~a consequence of this change, in~the current protocol, the~choice of which out of the two observables clients measures in each round is different, as~described~below. 

Consider two orthonormal qubit states $|0\rangle$ and $|1\rangle$ of the computational basis $\mathcal{B}_{+}$. The~diagonal basis $\mathcal{B}_{\times} = \{ |-\rangle, |+\rangle \}$ is given by $|\pm\rangle=\frac{1}{\sqrt{2}}(|1\rangle\pm |0\rangle)$. We also define two measurement observables:
\begin{equation}
\begin{array}{rl}
\hat{O}_{+} = & 1\cdot |1\rangle\langle 1| + 0\cdot |0\rangle\langle 0|, \vspace{1.7mm}\\
\hat{O}_{\times} = & 1\cdot |+\rangle\langle +| + 0\cdot |-\rangle\langle -|. 
\end{array}
\end{equation}
Let the bit string $M$ be the contract upon which Alice and Bob agree. Let $h$ be a publicly known hash function that they also agree to use. The~value $h^\ast=h(M)$ will be the only information they provide to Trent about the contract $M$, when and if they contact~him. 

The protocol is described below in three parts: (i) initialization phase, Stage I: Alice and Bob contact the trusted party, Trent, who provides each of them with different secret classical information, to~be used in the later phases of the protocol; Stage II: Trent prepares and distributes the states among the clients, to~be used to sign the contract; (ii) exchange phase: the clients, using the states and the information provided by Trent, ping-pong the information needed to sign the contract; and (iii) binding phase: in this phase of the protocol, any one of the clients can contact Trent with his/her results in order to obtain an authorized document declaring the hash value $h^\ast=h(M)$ valid, which then validates the contract $M$.

Note that the exchange of classical information between Trent and the clients, in~Stage I of the initialization phase, occurs over a private channel. During~the rest of the protocol, the~exchange of both quantum and classical information between Trent and the clients only needs to be authenticated. Classical authentication can be dealt with in a similar manner as in QKD, that is by assuming that a secret key is shared between Trent and the clients. How this key is exchanged depends on the level of security that we want to achieve: either an initial key is pre-shared (using a private channel), and~then, it is enlarged using QKD (information-theoretical security); or a public key infrastructure is used (computational security). It is relatively easy to authenticate quantum information upon having classical authentication. It reduces to applying the cut-and-choose technique to verify the authenticity of quantum states exchanged, i.e.,~some random states are used by the clients and Trent to check whether the quantum channel was tampered with. This way, Trent discloses the description of the states over the authenticated classical channel, and~the clients check if what they received is according to what was~expected.

Below, we present the detailed description of the~protocol.

\vspace{12pt}
\noindent\textbf{Initialization phase: Stage I:}
\vspace{12pt}

\textbf{\em Parties involved:} Alice, Bob, and~Trent.

\textbf{\em Input:} Bit strings $k_A$ and $k_B$ of length $4N$ each and $2N$ randomly chosen indices from the set of indices of the $4N$ rounds of the protocol, $\mathcal{I} = \{1,2,\dots,4N\}$.

\textbf{\em Communication channel:} Private classical channel between Trent and the~clients.\\

Stage I of the initialization phase consists of the following~steps:
\begin{itemize}
\item[1.] Alice and Bob contact Trent for his services and inform him about the future time, at~which Trent will begin Stage II of the initialization~phase.

\item[2.] Trent provides Alice and Bob with randomly-generated bit strings $k_A$ and $k_B$, respectively, of~length $4N$ each. Alice and Bob prepare the strings $H_{\texttt A} = h(M) \oplus k_A$ and $H_{\texttt B} = h(M) \oplus k_B$ (bit-wise XOR), respectively. Note that in order to exchange the commitments to the whole string, we have $4N = |h(M)|$. We define the honest observables $\hat{H}_{\texttt{A}_i}$ and $\hat{H}_{\texttt{B}_i}$ for Alice and Bob, respectively, to~be measured at each step $i$ of the protocol ($i \in \mathcal{I}$) as:
\begin{equation}
\hat{H}_{\texttt{A}_i} =
\begin{cases}
\hat{O}_{+} \:\:\:\:\:\:\:\: & \mbox{if } H_{\texttt{A}_i}= 0, \\
\hat{O}_{\times} \:\:\:\:\:\:\:\: & \mbox{if } H_{\texttt{A}_i} = 1,
\end{cases}
\end{equation}
\begin{equation}
\hat{H}_{\texttt{B}_i} =
\begin{cases}
\hat{O}_{+} \:\:\:\:\:\:\:\: & \mbox{if } H_{\texttt{B}_i}= 0, \\
\hat{O}_{\times} \:\:\:\:\:\:\:\: & \mbox{if } H_{\texttt{B}_i} = 1,
\end{cases}
\end{equation}
where $H_{\texttt{A}_i}$ is the 
$i-{\text{th}}$ bit of the string $H_{\texttt{A}}$ and analogously for Bob's string $H_{\texttt{B}}$. Note that the secret keys $k_A$ and $k_B$ are used by the clients to hide their respective honest observables from each other. As~it will be clear later on, the~introduction of these keys prevents the scenario of a dishonest client, say, Bob, attacking the quantum channel between Trent and Alice by measuring the correct honest observables on Alice's qubits to obtain perfect~correlations.

\item[3.] Trent provides Alice with a set of 2$N$ randomly-chosen indices from the total 4N indices, $\mathcal{I}^{(A)} \subset \mathcal{I}$. Analogously, he randomly chooses 2$N$ indices $\mathcal{I}^{(B)} \subset \mathcal{I}$ and sends them to Bob.

\item[4.] Trent provides Alice with 2$N$ bits of Bob's secret string $k_B$, corresponding to the above-mentioned 2$N$ indices from $\mathcal{I}^{(A)}$. Analogously, he sends the $2N$ bits to Bob from Alice's secret string $k_A$.
\end{itemize}

\noindent\textbf{\em The initialization phase: Stage I ends with the following:} \vspace{-2mm}
\begin{itemize}
\item Alice has a $4N$-long secret bit string $k_A$ and a set $\mathcal{I}^{(A)} \subset \mathcal{I}$ of $2N$ randomly-chosen indices from $\mathcal{I}$. Alice also has $k_{B_j}$ for all indices $j\in \mathcal{I}^{(A)}$.\vspace{-2mm}

\item Bob has a $4N$-long secret bit string $k_B$ and a set $\mathcal{I}^{(B)} \subset \mathcal{I}$ of $2N$ randomly-chosen indices from $\mathcal{I}$. Bob also has $k_{A_j}$ for all indices $j\in \mathcal{I}^{(B)}$

\end{itemize}

\noindent\textbf{Initialization phase: Stage II:}
\vspace{12pt}

\textbf{\em Parties involved:} Alice, Bob, and~Trent.

\textbf{\em Input:} $8N$ number of entangled~pairs.

\textbf{\em Communication channel:} Authenticated classical and quantum channels between Trent and the~clients.\vspace{12pt}

The initialization phase: Stage II consists of the following steps (see Figure~\ref{fig:initialization}):
\begin{itemize}
\item[5.] Trent produces two ordered sets, each consisting of $4N$ entangled pairs ($8N$ pairs in total). Each~pair of particles is in the state $|\phi^+\rangle=\frac{1}{\sqrt{2}}(|0\rangle|0\rangle + |1\rangle|1\rangle)$. He sends one particle from each pair of the first set to Alice, and~from the second set to Bob, keeping the order of the pairs preserved. Let us denote the ordered set of $4N$ particles given to Alice by $\mathcal{A}$ and those given to Bob by $\mathcal{B}$. The~two ordered sets kept by Trent, each consisting of $4N$ particles entangled with particles sent to Alice and Bob, are denoted by $\mathcal{T}^{(A)}$ and $\mathcal{T}^{(B)}$, respectively. We would like to note that the use of ordered sets was previously proposed in~\cite{long:02}, later called the block transmission technique, crucial to quantum secure direct~communication.

\item[6.] According to the set of $2N$ indices sent to Bob, $\mathcal{I}^{(B)}$, Trent divides $\mathcal{T}^{(A)}$ into two ordered subsets of $2N$ particles each: $\mathcal{T}_T^{(A)}$ and $\mathcal{T}_B^{(A)}$, with~the $2N$ indices $j \in \mathcal{I}^{(B)}$ corresponding to the particles $\mathcal{T}_B^{(A)}$. Note that the original positions in $\mathcal{T}^{(A)}$ of each particle from $\mathcal{T}_T^{(A)}$ and $\mathcal{T}_B^{(A)}$ are preserved. In~other words, for~each particle from $\mathcal{T}_T^{(A)}$ and $\mathcal{T}_B^{(A)}$, Trent knows its position in $\mathcal{T}^{(A)}$ and, hence, with~which particle in $\alice$ it is entangled. The~same is done 
with particles from $\mathcal{T}^{(B)}$, entangled~with those in $\mathcal{B}$, obtaining $\mathcal{T}_T^{(B)}$ and $\mathcal{T}_A^{(B)}$.

\item[7.] Trent sends the ordered subsets $\mathcal{T}_B^{(A)}$ and $\mathcal{T}_A^{(B)}$ to Bob and Alice, respectively, each consisting of 2$N$ particles. The~particles in $\mathcal{T}_B^{(A)}$ and $\mathcal{T}_A^{(B)}$ are entangled with the corresponding particles ($2N$ of them) in the sets $\mathcal{A}$ and $\mathcal{B}$, given to Alice and Bob, respectively. Note that knowing the indices $\mathcal{I}^{(A)}$, Alice knows which particle from $\mathcal{T}_A^{(B)}$ is entangled with $2N$ of the particles in $\mathcal{B}$; and analogously for Bob.
Trent keeps the subsets $\mathcal{T}_T^{(A)}$ and $\mathcal{T}_T^{(B)}$ to himself, to~be used during the binding phase.
\end{itemize}
\begin{figure}[H]
\centering
\includegraphics[width=0.7\linewidth]{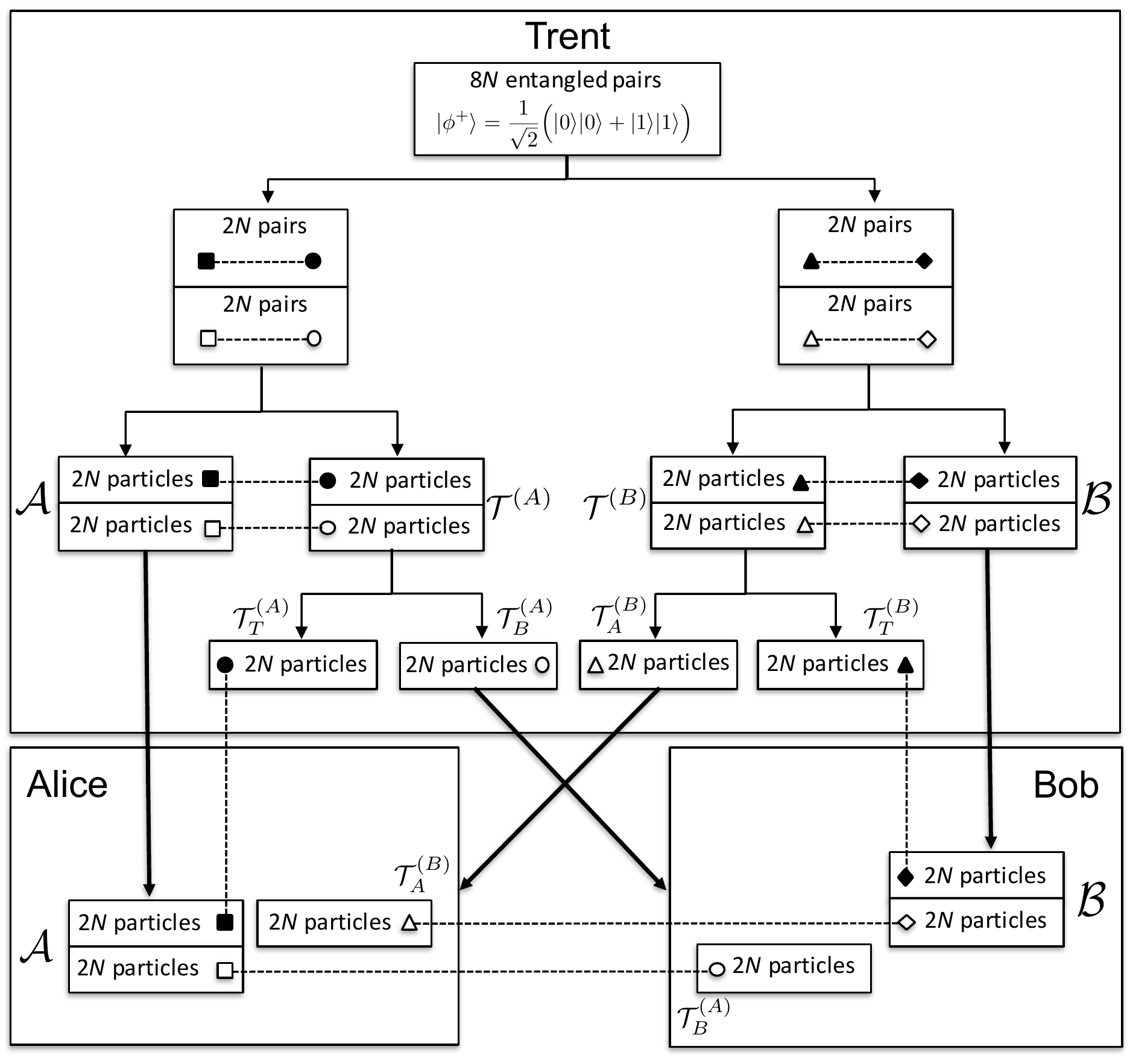}
\caption{{Initialization phase: Stage II: The thick arrows represent the transfer of particles from Trent to Alice and Bob. The~dashed connections represent the entanglement between the respective particles. The~big boxes represent the shielded private laboratories of Alice, Bob, and~Trent.}}
\label{fig:initialization}
\end{figure}

\noindent\textbf{\em The initialization phase: Stage II ends with the following:} \vspace{-2mm}
\begin{itemize}
\item Alice has an ordered set $\mathcal A$ of $4N$ particles, entangled with $2N$ particles kept by Trent, $\mathcal{T}_T^{(A)}$, and~additional $2N$ particles that are given to Bob, $\mathcal{T}_B^{(A)}$. She has another ordered set $\mathcal{T}_A^{(B)}$ of $2N$ particles, entangled with $2N$ particles given to Bob, chosen from $\mathcal B$ according to $\mathcal{I}^{(A)}$.\vspace{-2mm}

\item Bob has an ordered set $\mathcal B$ of $4N$ particles, entangled with $2N$ particles kept by Trent, $\mathcal{T}_T^{(B)}$, and~additional $2N$ particles that are given to Alice, $\mathcal{T}_A^{(B)}$. He has another ordered set $\mathcal{T}_B^{(A)}$ of $2N$ particles, entangled with $2N$ particles given to Alice, chosen from $\mathcal A$ according to $\mathcal{I}^{(B)}$.\vspace{-2mm}

\item Trent keeps two ordered sets of $2N$ particles each, $\mathcal{T}_T^{(A)}$ and $\mathcal{T}_T^{(B)}$, entangled with $2N$ particles from $\mathcal A$ and $\mathcal B$, respectively.
\end{itemize}

\vspace{6pt}
\noindent\textbf{Exchange phase:}\vspace{12pt}

\textbf{\em Parties involved:} Alice and~Bob.

\textbf{\em Input:}\vspace{-2mm}
\begin{itemize}
\item The particles and indices Alice and Bob obtained at the end of the initialization~phase.\vspace{-2mm}

\item $h(M)$, the~hash value of the contract $M$ to be signed, obtained using publicly known function $h$.\vspace{-2mm}

\item $H_{\texttt A} = h(M) \oplus k_A$ and $H_{\texttt B} = h(M) \oplus k_B$ for Alice and Bob, respectively.\vspace{-1.8mm}
\end{itemize}

\textbf{\em Communication channel:} Unauthenticated classical channel between the~clients.\vspace{12pt}

The exchange phase (see Figure~\ref{fig:exchange}) consists of $4N$ rounds. In~each round, a~client, say Alice, has a particle from $\mathcal A$, on~which she measures $\hat{H}_{\texttt{A}_i}$, with~$i\in \mathcal{I}$, and~sends the results to Bob. In~addition to this, in~$2N$ rounds labeled by $j\in\mathcal{I}^{(A)}$, Alice measures $\hat{H}_{\texttt{B}_j}$ on the corresponding particles from $\mathcal{T}_A^{(B)}$. Note that Alice knows $k_{\texttt{B}_j}$, and~therefore, she knows $H_{\texttt{B}_j} = h(M) \oplus k_{\texttt{B}_j}$. Since Alice 

knows that those particles are entangled with Bob, she uses her measurement outcomes to compare them with the results received from Bob, thus checking his behavior. Bob performs his measurements analogously. These two kinds of measurements are shown in Figure~\ref{fig:measurements}.
\begin{figure}[H]
\centering
\includegraphics[width=0.78\linewidth]{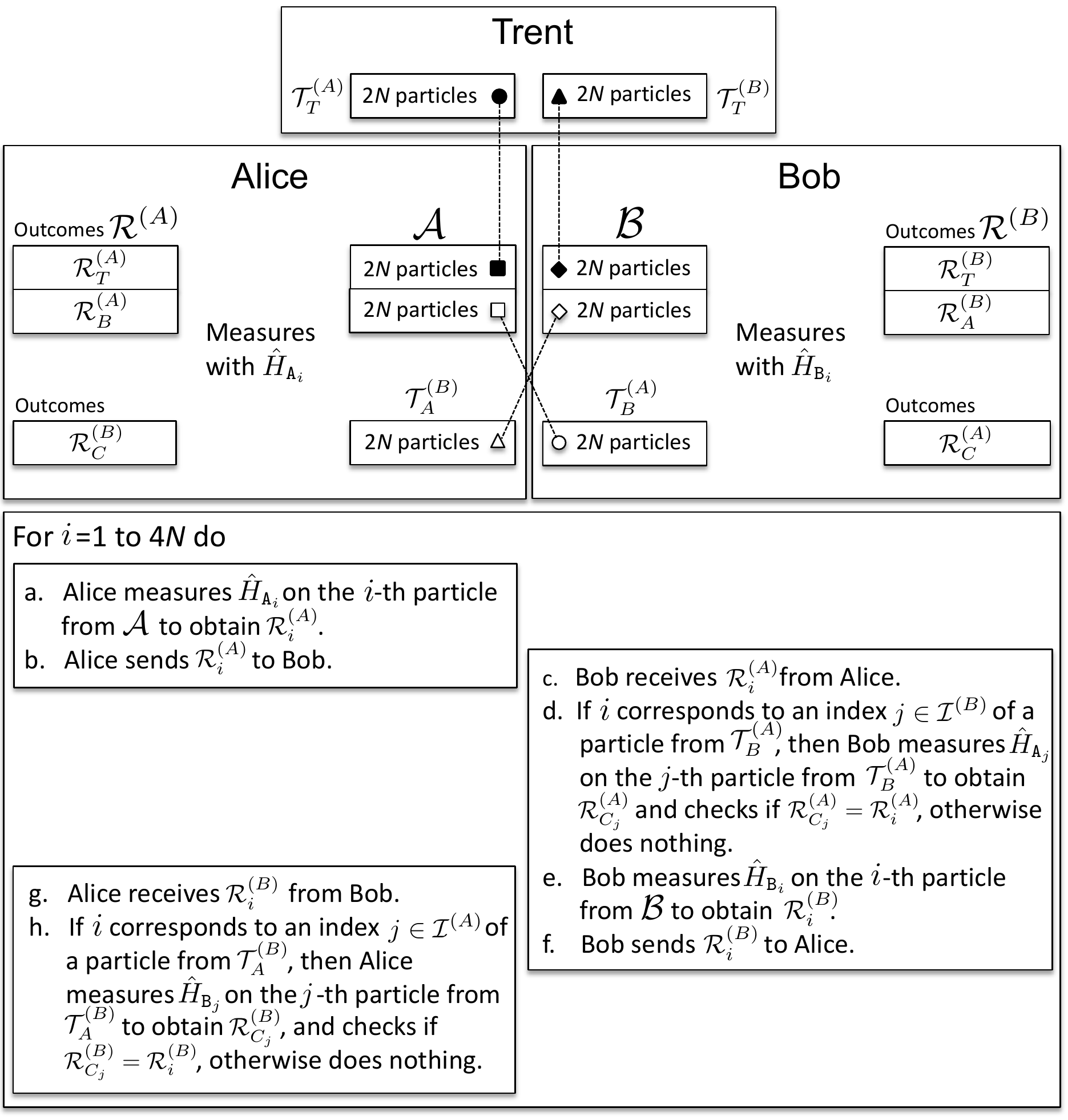}
\caption{{Exchange phase: the dashed connections represent the entanglement between the respective particles. The~steps of the protocol during the exchange phase are described.}}
\label{fig:exchange}
\end{figure}

\begin{itemize}
\item[8.] At the beginning of the exchange phase, Alice and Bob are in possession of $6N$ particles each. Alice has $4N$ particles denoted by $\mathcal A$ and $2N$ particles by $\mathcal{T}_A^{(B)}$, and~analogously for Bob. On~the $4N$ particles from $\mathcal A$, Alice measures her honest observable $\hat{H}_{\texttt{A}_i}$, with~$i = 1, \dots , 4N$. Bob measures $\hat{H}_{\texttt{B}_i}$ on his corresponding particles from $\mathcal B$. Their measurement outcomes form ordered sets of binary results $\mathcal{R}^{(A)}$ and $\mathcal{R}^{(B)},$ respectively, where:
\begin{equation}
\begin{array}{rl}
\mathcal{R}^{(A)} =& \mathcal{R}_T^{(A)} \cup \mathcal{R}_B^{(A)}, \vspace{1.7mm}\\
\mathcal{R}^{(B)} =& \mathcal{R}_T^{(B)} \cup \mathcal{R}_A^{(B)}.
\end{array}
\end{equation}
We use $\mathcal{R}_T^{(A)}$ to denote Alice's measurement results on the particles from $\alice$ ($2N$ of them) that are entangled with $\mathcal{T}_T^{(A)}$ (kept by Trent) and $\mathcal{R}_B^{(A)}$ to denote her measurement results on the rest of the particles from $\alice$ ($2N$ of them) entangled with $\mathcal{T}_{B}^{(A)}$ (given to Bob); and analogously for Bob's results, $\mathcal{R}_T^{(B)}$ and $\mathcal{R}_A^{(B)}$. They send these results to each other, one-by-one: if Alice starts first: in the $i-{\text{th}}$ step of the exchange, she sends to Bob her result $\mathcal{R}_i^{(A)}$, then Bob sends to Alice his result $\mathcal{R}_i^{(B)}$, and~so~on.
\begin{figure}[H]
\centering
\includegraphics[width=0.5\linewidth]{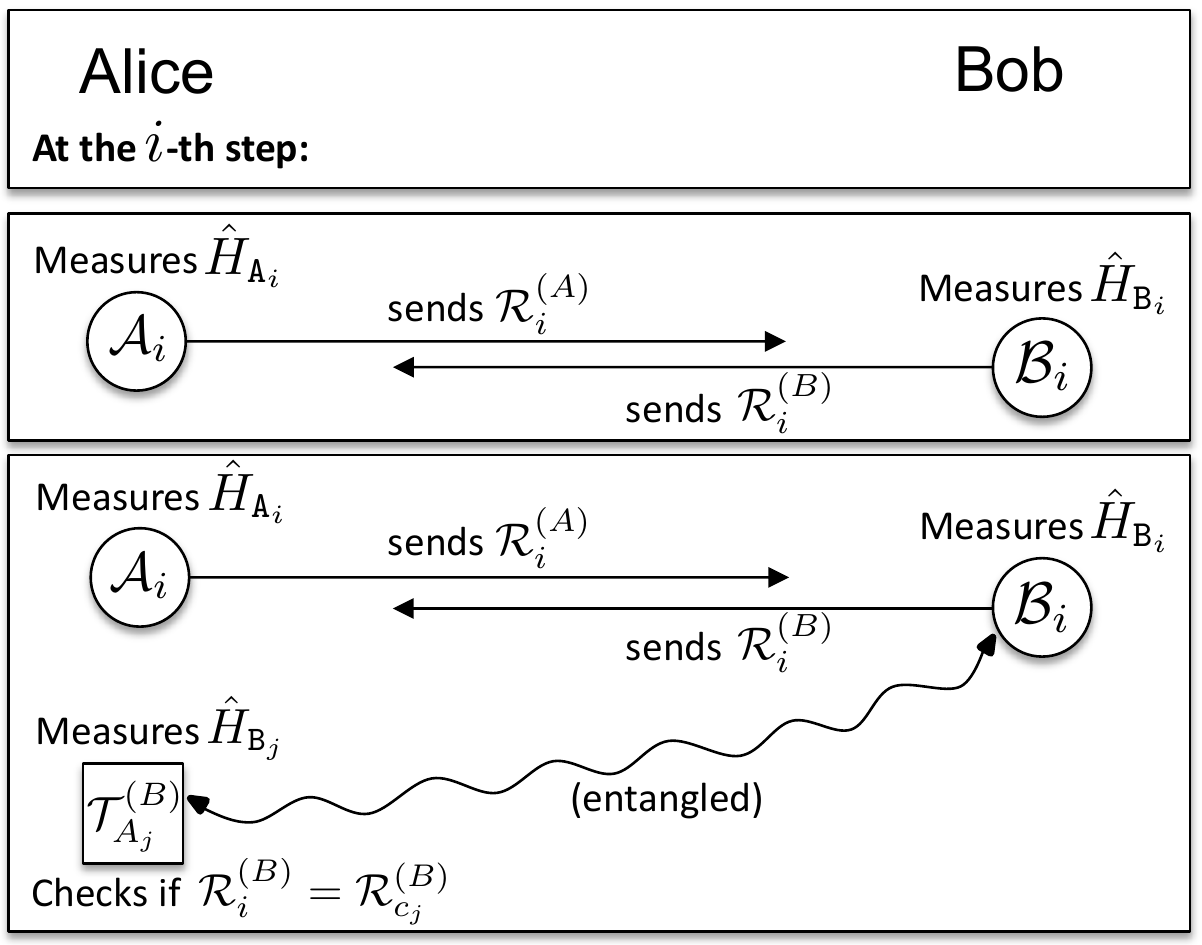}
\caption{{The two kinds of measurements a client, say, Alice, performs at a step $i$.}}
\label{fig:measurements}
\end{figure}

\item[9.] For each round $i\in \mathcal I$ for which there exists $j\in \mathcal{I}^{(A)}$, such that $i=j$, Alice measures $\hat{H}_{\texttt{A}_j}$ on the corresponding particle from $\mathcal{T}_A^{(B)}$, to~obtain $\mathcal{R}_{C_j}^{(B)}$ (see Figure~\ref{fig:measurements}). If~Bob indeed measured his honest observable $\hat{H}_{\texttt{B}_i}$, then his measurement outcome will match Alice's, $\mathcal{R}_i^{(B)}=\mathcal{R}_{C_j}^{(B)}$. In~the presence of noise, Alice applies a statistical test to verify if Bob provided enough consistent results (see Appendix~\ref{sec:honest_noisy}). In~case $i \neq j \in \mathcal{I}^{(A)}$, Alice uses Bob's result $\mathcal{R}_i^{(B)} \in \mathcal{R}_T^{(B)}$ for the optional binding phase, when Trent confronts Alice's information about $\mathcal{R}_i^{(B)}$ with his own measurement outcomes. The~same is done by Bob upon receiving $\mathcal{R}_i^{(A)}$ from Alice. Then: (i) if all measurement outcomes, $\mathcal{R}_C^{(A/B)}$ ($2N$ of them for each client), are found to be consistent by the end of the communication at step $4N$, both clients will, during~the binding phase, obtain with certainty the certified document from Trent that allows them to acquire a signed contract from the authorities (see the description below of the binding phase); (ii) if one of the clients suspects dishonest behavior, the~communication is stopped, and~they measure their honest observables on all remaining qubits and proceed to the binding phase.
\end{itemize}

\noindent\textbf{\em The exchange phase ends with the following:} If no cheating occurred, Alice and Bob both obtain their own, as~well as each others' results, $\mathcal{R}^{(A)}$ and $\mathcal{R}^{(B)}$. In~case the communication was interrupted at step $m$, a~client, say, Alice, ends up with all of her own results $\mathcal{R}^{(A)}$ and those received from Bob by the step $m$ (note that those do not necessarily need to be obtained by actually performing measurements on~qubits).

\vspace{12pt}
\noindent\textbf{Binding phase:}
\vspace{12pt}

\textbf{\em Parties involved:} Trent and a client, say, Alice.

\textbf{\em Input:} \vspace{-2.5mm}
\begin{itemize}
\item The sets $\mathcal{T}_T^{(A)}$ and $\mathcal{T}_T^{(B)}$ of particles kept by Trent.\vspace{-2mm}
\item The sets of Alice's measurement results, $\mathcal{R}^{(A)}$, and~those sent to Alice by Bob, $\mathcal{R}^{(B)}$. Note that in the case of Bob's cheating, $\mathcal{R}^{(B)}$ might contain the wrong values, and~in case the communication was interrupted at step $m$, it is only a partial set of results. For~simplicity, we use the same symbol $\mathcal{R}^{(B)}$ for both sets of ``honest'', as~well as ``dishonest'' results.\vspace{-2mm}
\item $h^\ast=h(M)$, the~hash value of the contract $M$ to be signed, obtained using publicly-known function~$h$.\vspace{-2mm}
\item A publicly-known distribution $p(\alpha)$ to choose the acceptance rate $\alpha$.\vspace{-1.8mm}
\end{itemize}

\textbf{\em Communication channel:} Private classical channel or in~person.\vspace{12pt}

\begin{figure}[H]
\centering
\includegraphics[width=0.67\linewidth]{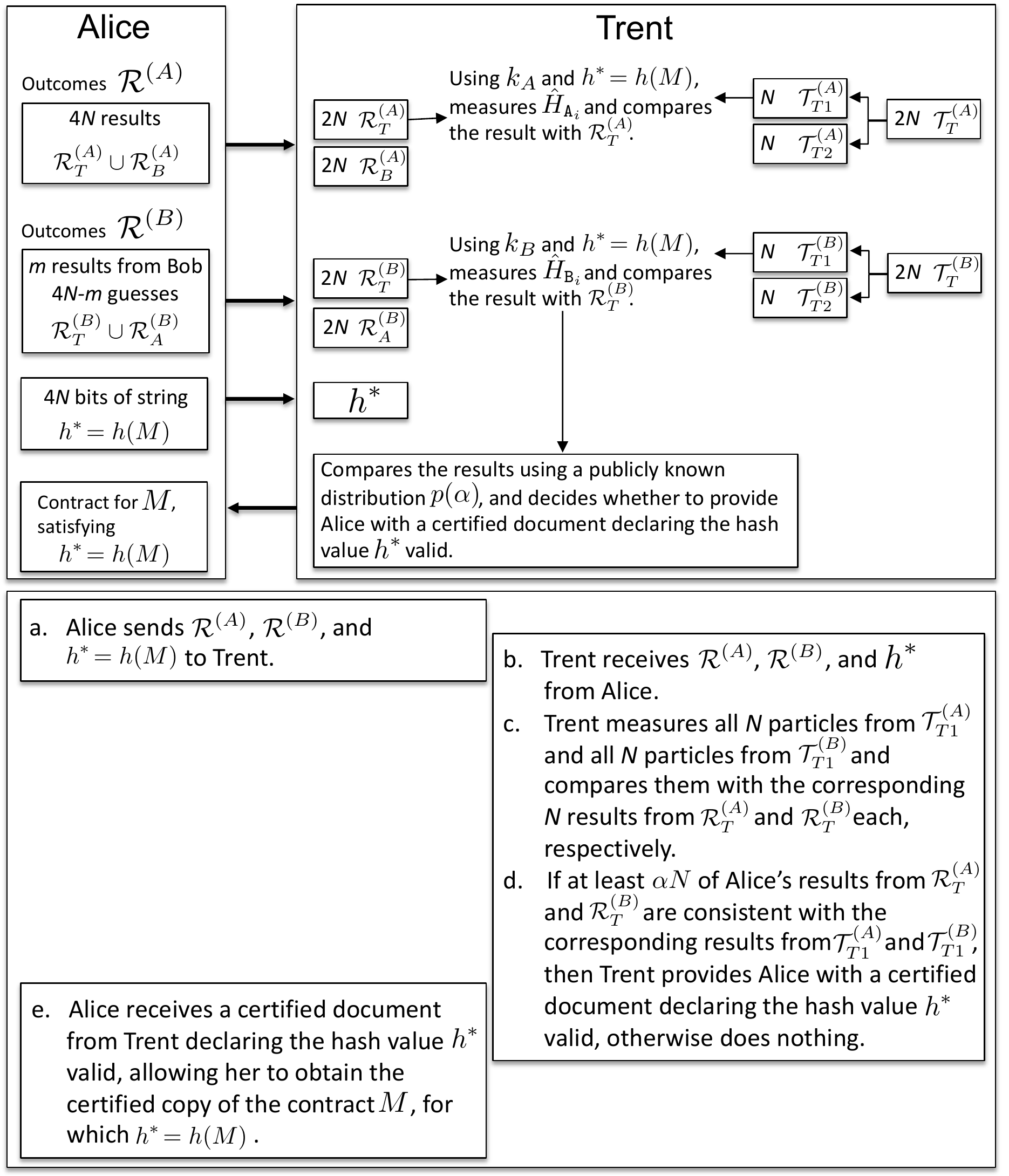}
\caption{{Binding phase: the thick arrows represent the transfer of measurement outcomes from Alice to Trent. The~steps for the binding phase are described.}}
\label{fig:binding}
\end{figure}
The binding phase (see Figure~\ref{fig:binding}) consists of the following step:
\begin{itemize}
\item[10.] During the binding phase, a~{single} client, say Alice, presents her results to Trent in order to bind the contract, to~receive a certified document, signed by Trent's public key, declaring valid the hash value $h^\ast = h(M)$. Having such a certificate, Alice can appeal to the authorities to enforce the terms of the contract~$M$: she presents the contract $M$ and the signed document declaring the value $h^\ast$ valid, so that the authorities can verify that indeed $h^\ast = h(M)$ (note that the function $h$ is publicly known). As~pointed out in the Introduction, Trent is an agency accredited by the State (e.g., public~notaries). To~bind the contract, Alice provides the string $h^\ast = h(M)$ to Trent and presents her results $\mathcal{R}^{(A)}$, as~well as those obtained from Bob, $\mathcal{R}^{(B)}$ (if the protocol was interrupted before its completion, Alice will guess the rest of Bob's outcomes; see Appendix for details). Knowing $h^*$, Trent computes $H_{\texttt{A}_j} = h^\ast \oplus k_{\texttt{A}_j}$ and $H_{\texttt{B}_j} = h^\ast \oplus k_{\texttt{B}_j}$. Trent thus measures the honest observables $\hat{H}_{\texttt{A}_i}$ on $N$ randomly-chosen qubits from the subset $\mathcal{T}_T^{(A)}$ and $\hat{H}_{\texttt{B}_i}$ on another $N$ randomly-chosen qubits from $\mathcal{T}_T^{(B)}$ (the other $2N$ particles are kept for binding the contract for Bob, if~requested). He also chooses independently at random $\alpha > 1/2$, according to some publicly-known distribution $p(\alpha)$. Trent will give to Alice a certified document declaring the hash value $h^\ast$ valid ({``bind the contract for Alice''}) if the results $\mathcal{R}^{(A)}$ and $\mathcal{R}^{(B)}$ satisfy the following two~conditions:
\begin{itemize}
\item[10a.] at least a fraction $\alpha$ of $N$ Alice's results from $\mathcal{R}_T^{(A)}$ is equal to Trent's results on the corresponding entangled $N$ particles from $\mathcal{T}_T^{(A)}$, and
\item[10b.] at least a fraction $\alpha$ of $N$ Bob's results from $\mathcal{R}_T^{(B)}$ is equal to Trent's results on the corresponding entangled $N$ particles from $\mathcal{T}_T^{(B)}$.
\end{itemize}
\end{itemize}

\noindent\textbf{\em The binding phase ends with the following:} In case Trent finds Alice's results consistent with his, she receives an authorized document from him declaring the hash value $h^\ast$ valid, which then allows her to obtain the certified copy of the contract $M$, for~which $h^\ast = h(M)$.

\section{Security~Analysis}\label{sec:security}
In our protocol, a~cryptographic hash function $h$ is used to map contract $M$ to a bit string of fixed size $4N=|h(M)|$. Had Trent possessed an infinite computational power, he would be able to find the collisions, among~which one message would be the contract $M$. Nevertheless, the~problem of finding collisions for existing cryptographic hash functions (such as SHA256 and others) is not based on any mathematical or number theoretical conjecture, such as the hardness of factoring, but~on the fact that the hashing function is highly irregular and non-linear. Its security is at the same level of symmetric cryptography (such as AES), which is assumed to be beyond the capacity of quantum technologies to attack, and~moreover, AES is actually used in current commercial QKD services. Furthermore, note~that at the time this paper was written, it was not yet found a single collision for SHA256 of two meaningless texts, and~so, it is unforeseeable to find collisions for a given fixed text. Google used more than $9\times 10^{18}$ hashes to find a meaningless collision of SHA1~\cite{sha1}, and~SHA256 is considerably harder than SHA1. Finally, it is worthwhile noticing that the assumption that there exists an unbreakable hash function, the~so-called {random oracle model}, is quite common, even when quantum information and computation is available~\cite{QROM}. In~addition, having such computational power would also allow a cheating client (say, Bob) to find collisions as well, thus potentially giving him the opportunity to bind a different contract $M^\prime$, for~which $h(M^\prime) = h(M)$. Nevertheless, given a particular hash function $h$, it is negligible that other collisions different from the contract $M$ would still represent meaningful contracts, let alone contracts that would be favorable to~Bob.

Let us define the following probabilities for Alice to pass the above tests (a) and (b), in~case the communication was interrupted at step $m$,
\begin{eqnarray*}
P_{BTH}(m;\alpha) &-& \mbox{Probability that Bob passes Trent's test on his own qubits},\\
P_{BTA}(m;\alpha) &-& \mbox{Probability that Bob passes Trent's test on Alice's qubits},
\end{eqnarray*}
$P_{ATH}(m;\alpha)$ and $P_{ATB}(m;\alpha)$ can be analogously defined as Alice's probabilities to pass Tests (a) and (b). Additionally, we define Bob's probability to pass Alice's verification test on the results $\mathcal{R}_i^{(A)} \in \mathcal{R}_B^{(A)}$ (see Step (5) of Section~\ref{sec:setup}) received from Alice by step $m$, as:
\begin{eqnarray*}
P_{BAS}(m) &-& \mbox{Probability that Bob passes Alice's test on their shared qubits,}
\end{eqnarray*}
and analogously $P_{ABS}(m)$ for~Alice.

It is easy to verify that, in~the noiseless scenario, the~protocol is {optimistic}. 
If~both clients follow the protocol honestly till the end, both of them are able to enforce the contract: Alice will have all the consistent results for her own, as~well as Bob's qubits, allowing her to bind the contract with probability one (the same happens for Bob).

To analyze the probabilistic fairness quantitatively, we introduce the so-called {probability to cheat}, along the lines of the similar quantity analyzed in~\cite{paun:2011}. By~$P^\bob_{bind}(m;\alpha) = P_{BTH}(m;\alpha)\cdot P_{BTA}(m;\alpha)$, we~denote the probability that Bob passes Trent's tests and can thus bind the contract, if~the communication is interrupted at step $m$ of the protocol, for~a given choice of $\alpha$; and analogously for Alice. To~reach step $m$, both clients have to pass each others' verification, which is given by the probability $P(m) = P_{BAS}(m)\cdot P_{ABS}(m)$. Bob's probability to cheat at step $m$, for~a given $\alpha$, is defined as the probability that he can bind the contract, while Alice cannot, multiplied by the probability to reach step $m$:
\begin{equation}\label{eq:cheat1}
P_{ch}^\bob (m;\alpha)=P(m)\cdot P^\bob_{bind}(m;\alpha)\left[1-P^\alice_{bind}(m;\alpha)\right].
\end{equation}

Note that the above probabilities also depend on the particular distribution of entangled pairs, denoted as ``configuration $\mathcal L$'', given by probability $q(\mathcal L)$, and~in the case of a dishonest client, the~cheating strategy. Furthermore, both the above, as~well as {any} probability evaluated (with the exception of $p(\alpha)$) depend on $N$; therefore, we omit writing it, as~it is implicitly assumed. Nevertheless, the~dependence on configuration $\mathcal L$ is relevant in calculations, and~below, we analyze it in~detail.

As prescribed by the protocol, Trent gives $6N$ qubits to Alice: $4N$ qubits from $\alice$ and $2N$ from $\mathcal T^{(B)}_A$ (see Figure~\ref{fig:honest}), together with their relative positions. Analogously, Bob receives $6N$ qubits from $\bob$ and $\mathcal T^{(A)}_B$. We assume that all the classical communications between Trent and clients are private and authenticated, based on, say, pre-shared symmetric key schemes. After~the communication has stopped at step $m$, out of the $4N$ qubits to measure from $\alice$ and $\bob$, Alice and Bob will be left with $\ell^{(A)}$ and $\ell^{(B)}$ unmeasured qubits, respectively, with~$\ell^{(A)} = \ell^{(B)} = 4N - m$. Note that, among~the $\ell^{(A)}$ and $\ell^{(B)}$ qubits, not all will be used by Trent to bind the contract for Alice and Bob. In~fact, the~qubits that are entangled between Alice and Bob are irrelevant for their binding probabilities. They are used by the parties to check each others' honesty. Let then $\ell^{(A)}$ be decomposed into $\ell_{T}^{(A)}$ and $\ell_{B}^{(A)}$, the~qubits entangled with those held by Trent and by Bob, respectively. Analogously, let $\ell^{(B)}$ be decomposed into $\ell_{T}^{(B)}$ and $\ell_{A}^{(B)}$. Therefore:
\begin{equation}
\begin{array}{rrclcl}
\mbox{Alice } : & (2N+2N)-m & = & \ell^{(A)} & = & \ell_{T}^{(A)} + \ell_{B}^{(A)} \vspace{1.7mm}\\
\mbox{Bob } : & (2N+2N)-m & = & \ell^{(B)} & = & \ell_{T}^{(B)} + \ell_{A}^{(B)}.
\end{array}
\end{equation}

\begin{figure}[H]
\centering
\includegraphics[width=0.6\linewidth]{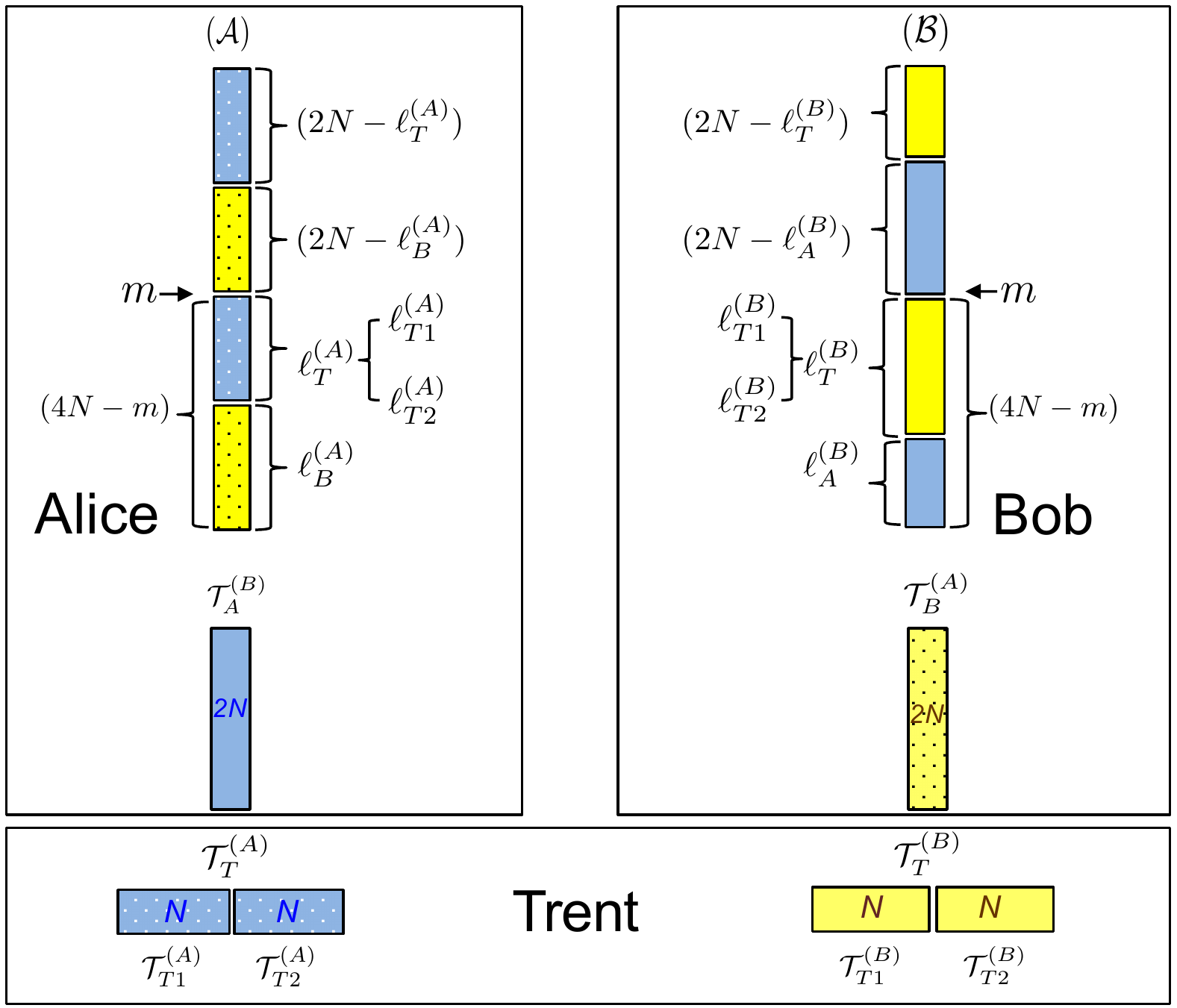}\\
\caption{{Decomposition of $\ell^{(A)}$ into $\ell_{T}^{(A)}$ and $\ell_{B}^{(A)}$ and decomposition of $\ell^{(B)}$ into $\ell_{T}^{(B)}$ and~$\ell_{A}^{(B)}$. The~same~colors~and~patterns~represent~entanglement~between~the~respective~particles~of~different~parties.}}
\label{fig:honest}
\end{figure}

In order to bind the contract, a~client, say Bob, has to present his own measurement results, as~well as those obtained from Alice. Then, Trent checks if they are correlated with those obtained on qubits in his possession. Unlike the previous proposal~\cite{paun:2011}, in~which {both} clients had to be present and show their results to Trent in order to both obtain signed contracts during the binding phase, in~the current protocol, Bob does not need Alice to be summoned in order to bind the contract (and~vice~versa). 
Since~the protocol should be symmetric to both clients, it should allow that they both, {separately}, are~able to bind the contract. For~this reason, when binding the contract to, say Bob, Trent does not check all of his results from $\mathcal T_T^{(A)}$ and $\mathcal T_T^{(B)}$ for qubits entangled with $\alice$ and $\bob$, respectively. Note~that to check Bob's results, Trent has to measure the honest observables $\hat{H}_{\texttt{A}_i}$ and $\hat{H}_{\texttt{B}_i}$ on his qubits, he~obtains using the $h^\ast=h(M)$ provided by Bob. Therefore, if~both clients were using the same sets of qubits (entangled with those in Trent's possession) to bind the contract separately, a~dishonest Bob would have a trivial successful cheating strategy. He measures his honest observables given by the mutually-agreed contract $M$, which allows him to bind that contract. Nevertheless, in~case he later decides not to comply with it, he simply provides Trent with a random $h'\neq h(M)$. As~a consequence, Trent's results will be uncorrelated with {both} Bob's, {as well as} Alice's results, i.e.,~neither client would be able to bind the contract $M$. This is precisely the reason for checking only $N$ out of $2N$ qubits from $\mathcal T_T^{(A)}$ and $\mathcal T_T^{(B)}$,~each.

Thus, Trent's qubits are each divided into two equal subsets of the same size, $\mathcal T_T^{(A)} = \mathcal T_{T1}^{(A)} \cup \mathcal T_{T2}^{(A)}$ and $\mathcal T_T^{(B)} = \mathcal T_{T1}^{(B)} \cup \mathcal T_{T2}^{(B)}$: the sets with the $T1$ subscript are used for binding the contract to Alice, while~those with $T2$ for Bob. Consequently, we have $\ell_{T}^{(A)} = \ell_{T1}^{(A)} + \ell_{T2}^{(A)}$ and $\ell_{T}^{(B)} = \ell_{T1}^{(B)} + \ell_{T2}^{(B)}$ (see~Figure~\ref{fig:honest}).

The overall configuration $\mathcal L$ of the entangled pairs distributed between Alice, Bob, and~Trent is given by six numbers,
\begin{equation}
\begin{array}{ll}
& \ \ \ \ \! \ \ \mathcal L = \left\{ \ell_{B}^{(A)}\!\!\!, \; \ell_{T1}^{(A)}\!, \; \ell_{T2}^{(A)}\!, \; \ell_{A}^{(B)}\!\!\!, \; \ell_{T1}^{(B)}\!, \; \ell_{T2}^{(B)} \right\},\vspace{2mm}\\ 
\mbox{with} \ \ \ \ \ \ \ \ &q(\mathcal L ) \leq q(\ell_{B}^{(A)})\! \cdot q(\ell_{T1}^{(A)})\! \cdot q(\ell_{T2}^{(A)})\! \cdot q(\ell_{A}^{(B)})\! \cdot q(\ell_{T1}^{(B)})\! \cdot q(\ell_{T2}^{(B)}).
\end{array}
\end{equation}

Therefore, Bob's probability to cheat, given by Equation~\eqref{eq:cheat1}, now written with the explicit dependence on the configuration $\mathcal L$, is:
\begin{equation}
\label{eq:cheat}
P_{ch}^\bob (m;\alpha|\mathcal L) = P(m|\ell_{B}^{(A)}\!\!\!,\ell_{A}^{(B)})\cdot P^\bob_{bind}(m;\alpha|\ell_{T2}^{(A)}\!,\ell_{T2}^{(B)}) \cdot \left[1-P^\alice_{bind}(m;\alpha|\ell_{T1}^{(A)}\!,\ell_{T1}^{(B)})\right].
\end{equation}

Averaging the ``constituent'' probabilities from the above equation over their respective configurations from $\mathcal L$ gives:
\begin{equation}
\begin{array}{ll}\label{eq:Pm}
&\quad \quad \ P(m) = \langle P_{ABS}(m|\ell_{B}^{(A)}) \rangle_{\ell^{(A)}_B}\! \cdot \langle P_{BAS}(m|\ell_{A}^{(B)}) \rangle_{\ell^{(B)}_A}, \vspace{1.7mm}\\
&P^\bob_{bind}(m;\alpha) = \langle P_{BTH}\!(m;\alpha|\ell_{T2}^{(B)}) \rangle_{\! \ell^{(B)}_{T2}}\! \cdot \langle P_{BTA}\!(m; \alpha|\ell_{T2}^{(A)}) \rangle_{\!\ell^{(A)}_{T2}},\vspace{1.7mm}\\
&P^\alice_{bind}(m;\alpha) = \langle P_{\!ATH}(m;\alpha|\ell_{T1}^{(A)}) \rangle_{\! \ell^{(A)}_{T1}}\! \cdot \langle P_{\!ATB}(m;\alpha|\ell_{T1}^{(B)}) \rangle_{\!\ell^{(B)}_{T1}},
\end{array}
\end{equation}
where $\langle A \rangle_\ell$ represents the expectation value of $A(\ell)$ over the values of $\ell$. To~simplify notation, in~the following, we will use $P_{ABS}(m) = \langle P_{ABS}(m|\ell_{B}^{(A)}) \rangle_{\ell^{(A)}_B}$ and $P_{BTH}(m;\alpha) = \langle P_{BTH}(m;\alpha|\ell_{T2}^{(B)}) \rangle_{\ell^{(B)}_{T2}}$, and~analogously for the other four probabilities from the right-hand sides of the above three~equations.

Hence, with~Bob's probability to cheat, averaged over all configurations $\mathcal L$,
\begin{eqnarray}
\label{cheat}
\!\!\!\!P_{ch}^\bob (m;\alpha)&=&\sum_{\mathcal L}q(\mathcal L) \;P_{ch}^\bob (m;\alpha|\mathcal L) \nonumber \\
&\leq&P(m)\cdot P^\bob_{bind}(m;\alpha) \left[1-P^\alice_{bind}(m;\alpha)\right],
\end{eqnarray}
we have the expected probability to cheat as:
\begin{equation}
\label{eq:final_cheat}
\bar{P}^\bob_{ch}(m) = \int p(\alpha)\: P_{ch}^\bob (m;\alpha)\: \mbox d \alpha .
\end{equation}

For honest clients that follow the protocol, the~above probability is determined by the steps prescribed by the protocol (the ``honest strategy''). In~case a client, say Bob, does not follow the protocol, the~above probability depends on the ``cheating strategy'' of a dishonest client. It turns out (see Appendix~\ref{sec:dishonest_noisy}) that the quantum part of the honest and the {optimal} cheating strategies is the same, i.e.,~the (quantum) measurements performed by a cheating Bob are the same as that of an honest one, given by his honest observables $\hat{H}_{\texttt{B}_i}$. In~other words, the~best a cheating Bob can do is to send to Alice the wrong results determined by a frequency $f$. This is a consequence of the fact that Bob does not know which of the qubits given to him are used to bind the contract by Trent and which to check his honesty by Alice (for details, see Appendix~\ref{sec:dishonest_noisy}).

In Appendix~\ref{sec:soundness}, we derive the explicit expressions for the expected probability to cheat~\eqref{eq:final_cheat} for honest clients that follow the protocol, in~the ideal noiseless case (Appendix~\ref{sec:honest_noiseless}), as~well as for noisy environments (Appendix~\ref{sec:honest_noisy}), thus showing the soundness of the protocol. In Figures \ref{fig:plots}a,b,d,e, we present the values of the maximal expected probability to cheat as a function of the total number of photons for the values of $4N$ up to 6000. In~both cases (as well as for the case of a cheating client discussed below), the~results are obtained for the uniform $p(\alpha)$ on the intervals $[0.9,0.99]$, $[0.8,0.99]$, and~$[0.7,0.99]$, and~with a noise parameter $\kappa = 0.05$. The~observed dependence $\max_m \bar{P}_{ch}^{\alice / \bob} (m) \propto N^{-1/2}$ is confirmed by the proof of the asymptotic behavior, $P_{ch}(m;\alpha) \in \mathcal O (N^{-1/2})$ (see Theorem~\ref{theorem} from Appendix~\ref{sec:honest_noiseless}). 
\begin{figure}[H]
\centering
\subfloat[Honest Noiseless]{\includegraphics[width=5cm]{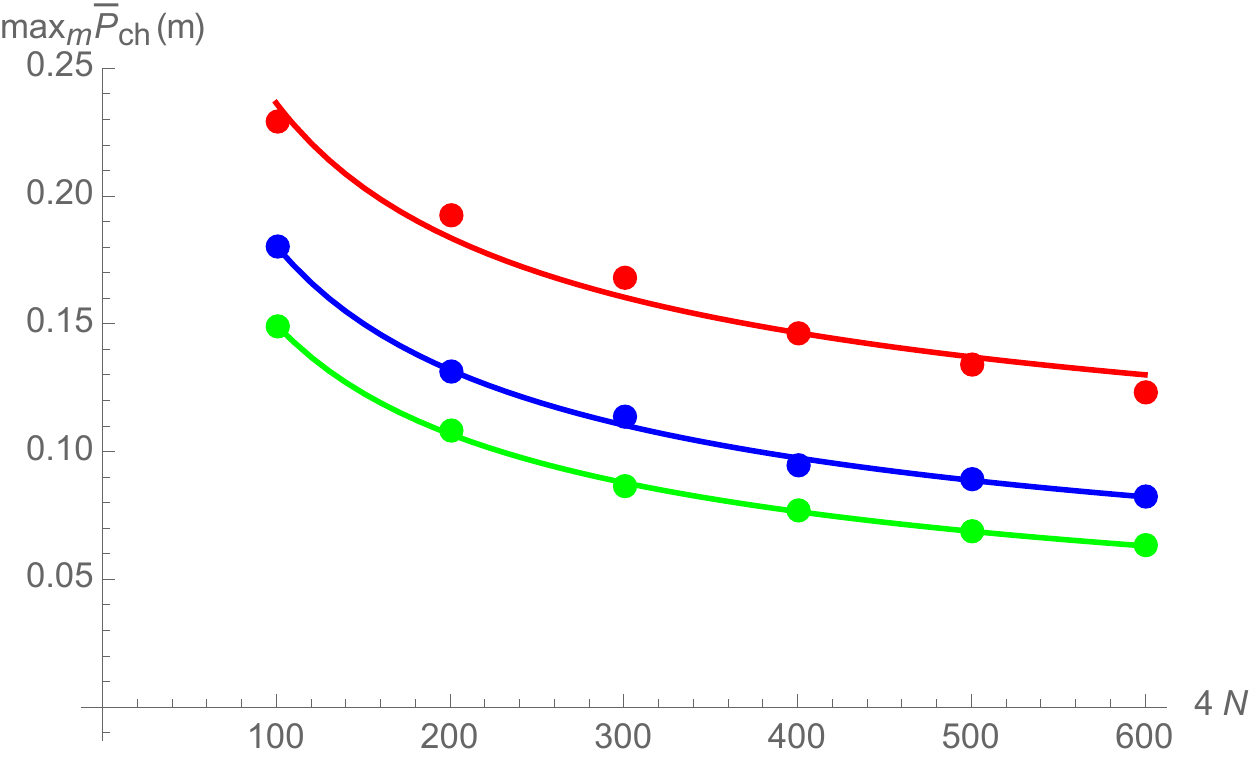}} 
\subfloat[Honest Noisy]{\includegraphics[width=5cm]{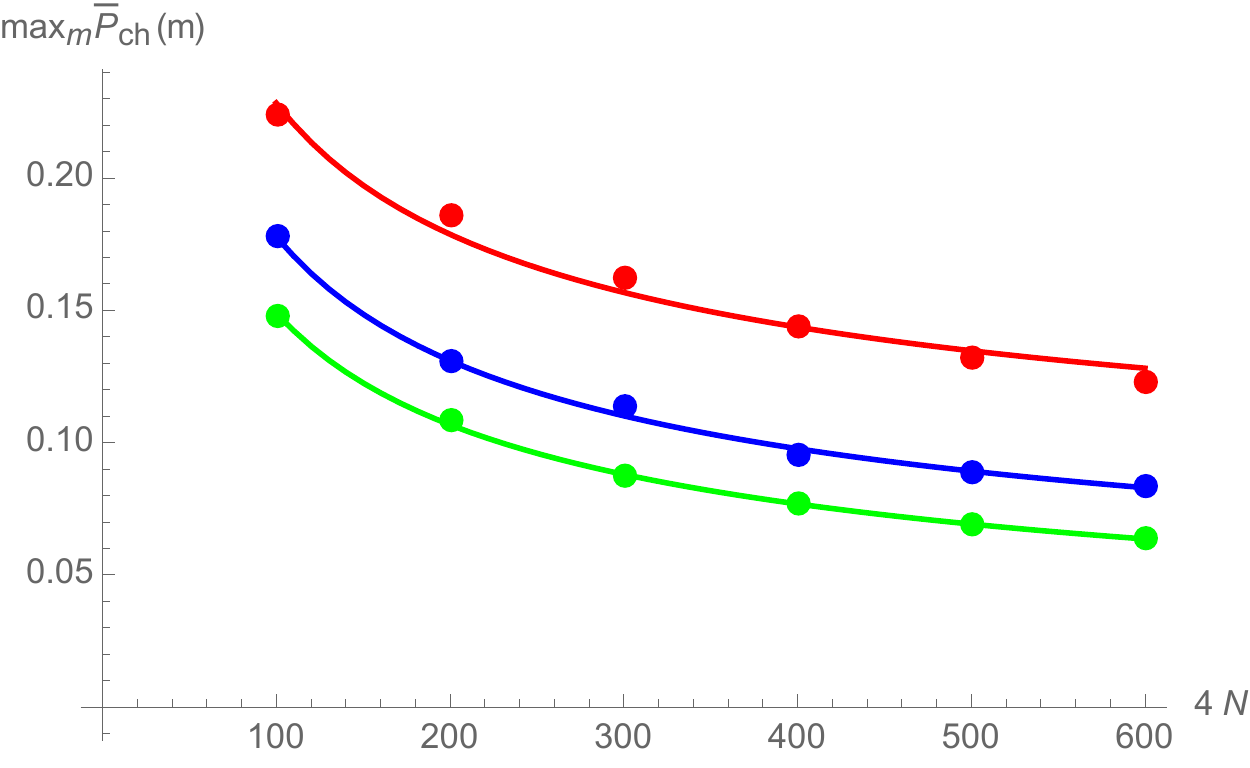}}
\subfloat[Dishonest Noisy]{\includegraphics[width=5cm]{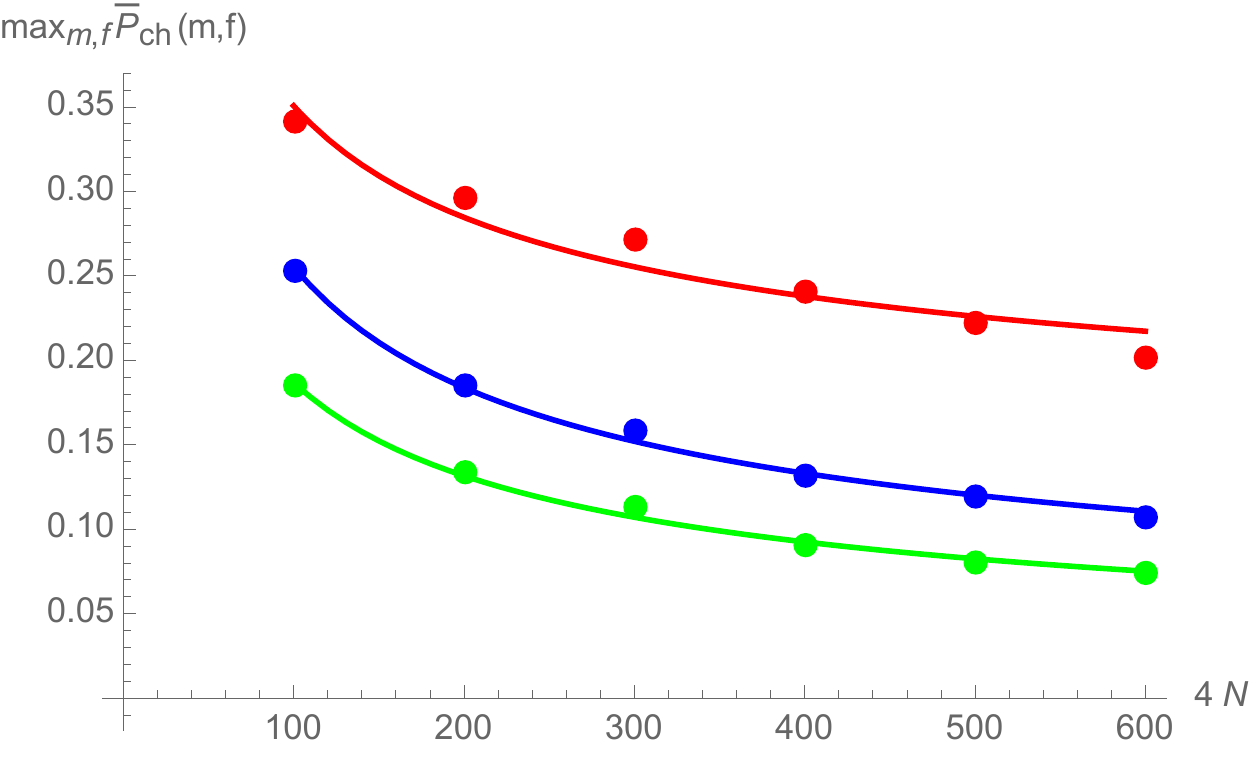}}\\
\subfloat[Honest Noiseless]{\includegraphics[width=5cm]{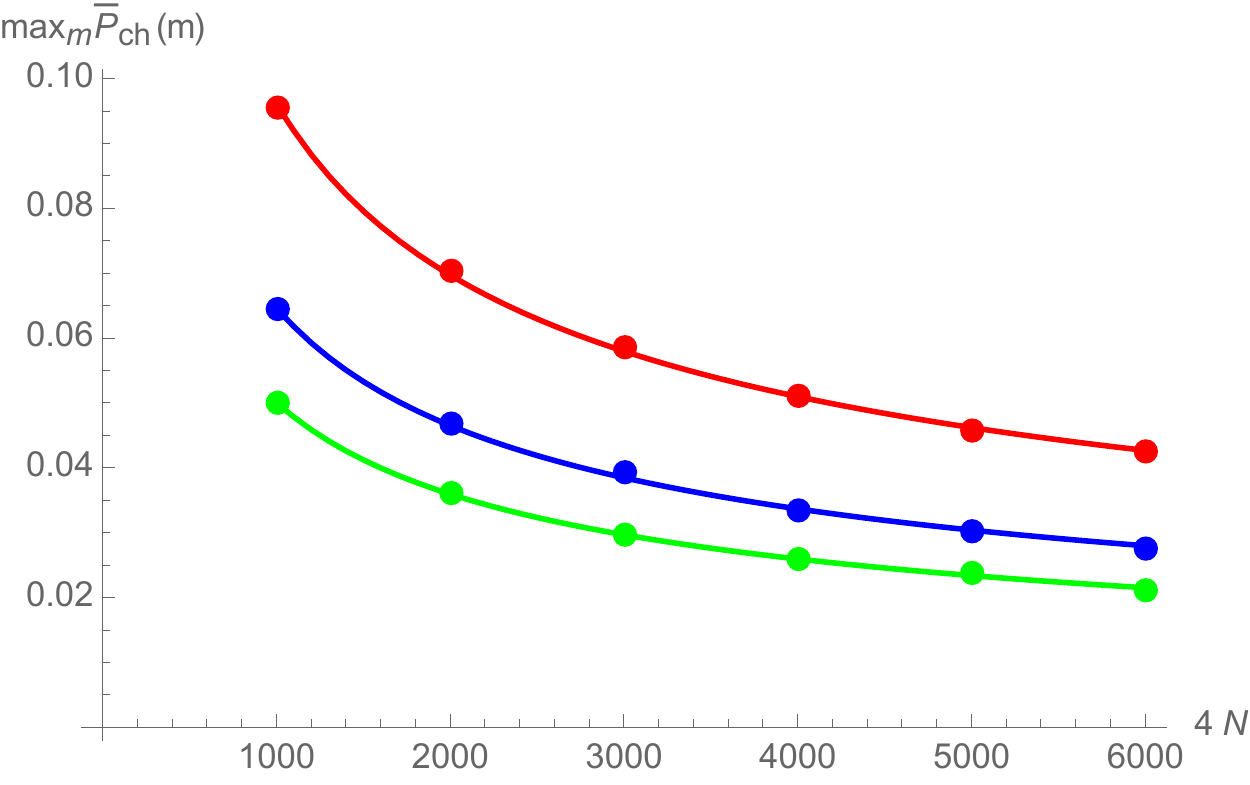}}
\subfloat[Honest Noisy]{\includegraphics[width=5cm]{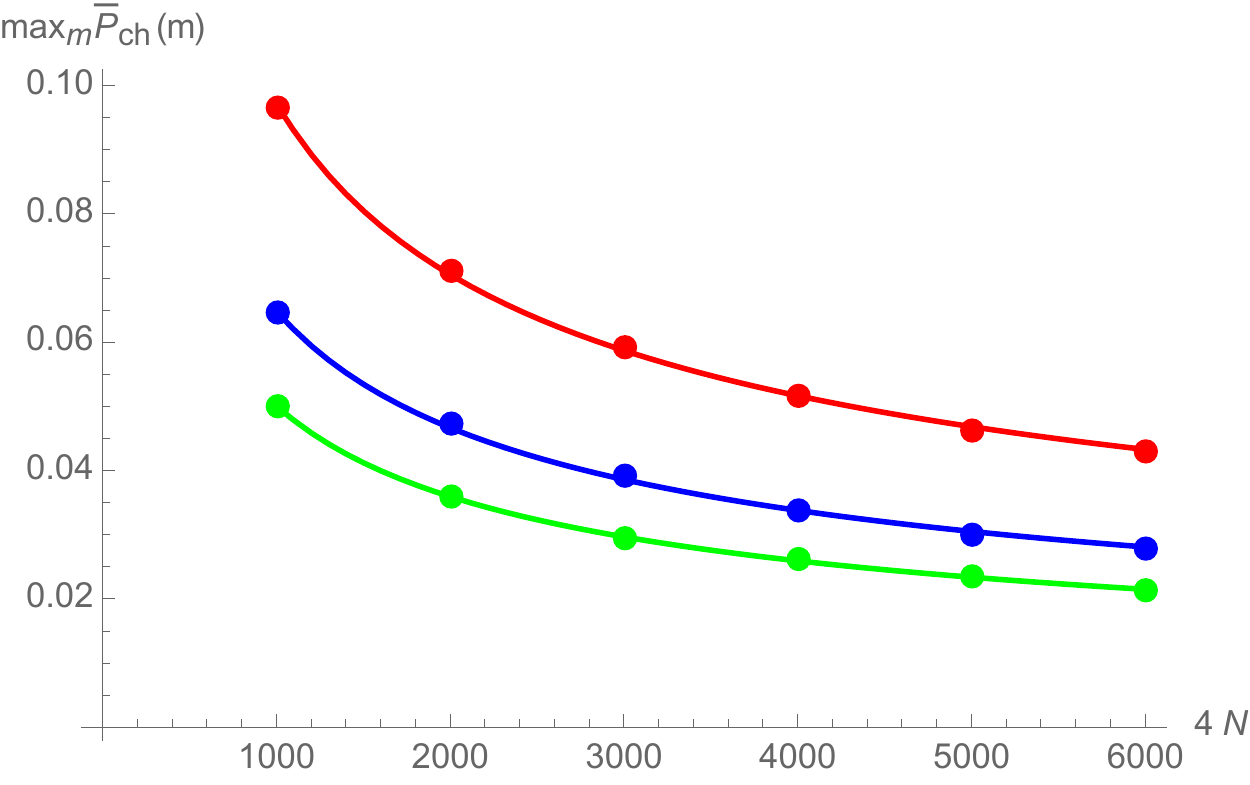}}
\subfloat[Dishonest Noisy]{\includegraphics[width=5cm]{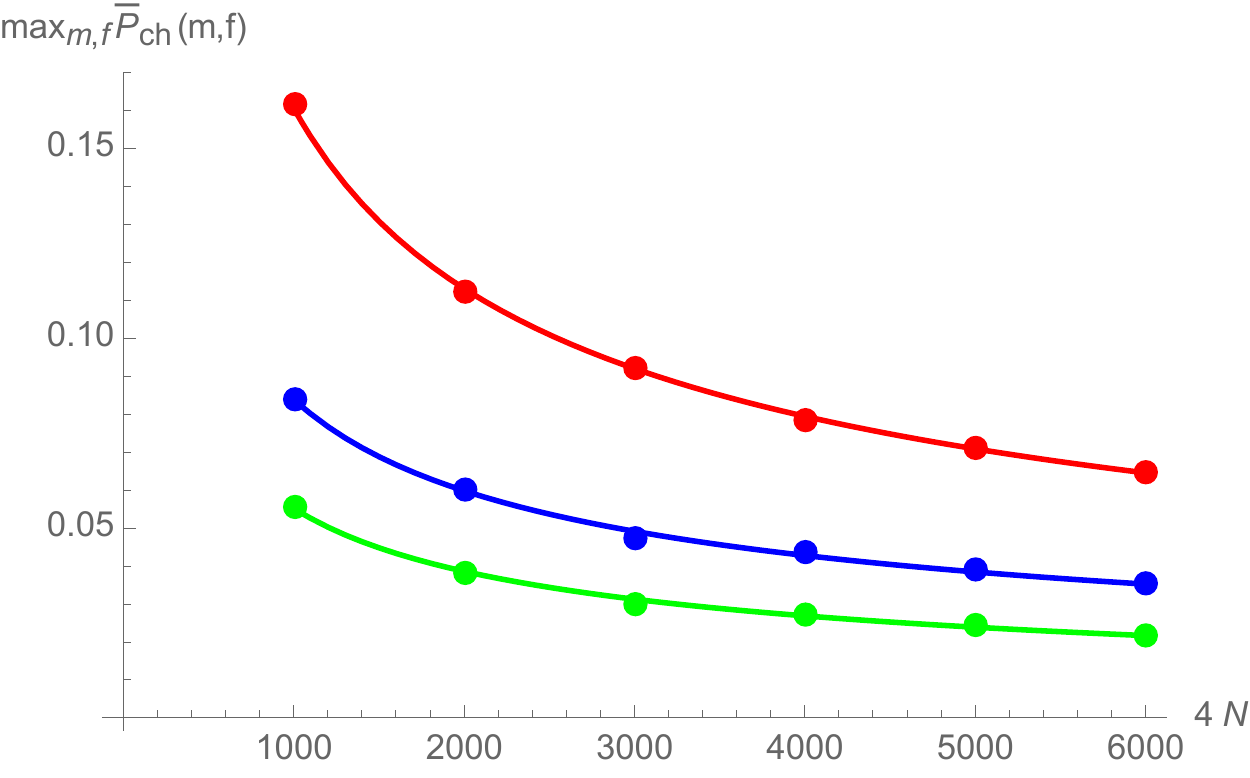}}
\caption{{(\textbf{a},\textbf{d}) correspond to the honest noiseless scenario, where the maximal expected probability to cheat, $\max_{m} \bar{P}_{ch}(m)$, is plotted against the total number of rounds between Alice and Bob, $4N$. (\textbf{b},\textbf{e})~correspond to the case of honest clients in a noisy channel with a noise parameter of $\kappa = 0.05$. (\textbf{c},\textbf{f})~correspond to the realistic case of a dishonest client with noisy channels, where the optimal cheating strategy depends on a parameter $f$. In~all three cases, the~red, blue, and~green curves are obtained for the uniform $p(\alpha)$ on the intervals $[0.9,0.99]$, $[0.8,0.99]$, and~$[0.7,0.99]$, respectively.}}
\label{fig:plots}
\end{figure}
In Appendix~\ref{sec:dishonest_noisy}, we evaluated the corresponding probabilities for the case of a cheating client who deviates from the protocol, in~the presence of noise. In~Figure~\ref{fig:plots}c,f, we plot the values of the maximal expected probability to cheat against $4N$, for~the optimal cheating strategy, showing the same dependence $\propto N^{-1/2}$. Further, in~Appendices~\ref{sec:soundness} and~\ref{sec:dishonest_noisy}, we analyze the decrease of the expected probabilities to cheat in case the cheating strategy deviates from the optimal values of $m$ or~$f$.

The results presented in Figure~\ref{fig:plots} are obtained for a fixed value of the noise parameter $\kappa$. By~increasing the noise, more and more ``wrong'' results are going to be obtained, such that even honest participants will either interrupt the communication during the exchange phase or will not be able to bind the contract with Trent. The~figure of merit here is the final average probability to bind obtained for $m=4N$, given by Equations~\eqref{eq:BTHBTA}--\eqref{eq:BTA} from Appendix~\ref{sec:honest_noisy} (note that, if~upon exchanging all the messages, clients have high enough probability to bind the contract, then they would also be able to pass each others' tests during the whole exchange phase with equally high probability). 

Thus, for~a fixed number of rounds, $4N$, range of $\alpha$, and~its probability distribution $p(\alpha)$, one can straightforwardly obtain the threshold values for the noise parameter $\kappa$, both for the honest noisy, as~well as for dishonest noisy cases (note that such threshold values depend on predetermined security level, i.e.,~the lower bound for the binding probability set by the users).

While such quantitative numbers can straightforwardly be obtained using the analysis presented in the paper, they would not be very informative. Namely, what one should do is to, {given the noise level} (given $\kappa$), determined by the actual implementation setup, optimize the rest of the relevant parameters ($N, p(\alpha)$ and its range). While conceptually, this is possible to do, it is clearly exceptionally demanding regarding the computational resources (note that in our analysis, we probed only three $\alpha$ ranges and only the simplest uniform distribution). Our paper is more of a proof of a concept, rather than the final analysis, which, as~mentioned, is strongly implementation~dependent.

Therefore, considering our limited computational power and the fact that the presented threshold values for $\kappa$ would probably differ from the ones to be obtained by optimizing the parameters of actual implementations, we decided to omit such numerical analysis. Note that the above discussion also applies for the cheating probability: to obtain the optimized cheating probability levels, one should vary all the relevant parameters. Nevertheless, while it is obvious that the binding probability will decrease as the noise increases, it was not at all obvious that it is even possible to establish upper bounds for the cheating probability, such that it can be made arbitrarily~low.

The techniques presented in our paper use the ``brute force'' numerical approach in obtaining the final quantitative results (with the exception of our analytic proof of the asymptotic behavior given in Theorem~\ref{theorem} from Appendix~\ref{sec:honest_noiseless}), which do not allow for drawing qualitative insights. Developing more closed analytic expressions for the final binding and cheating probabilities that can be analyzed beyond the final numerical values would be an interesting topic of future~research.

\section{Discussion} \label{sec:discussion}
Like the previous version of the quantum contract-signing protocol~\cite{paun:2011}, the~proposed protocol relies on long-term stable quantum memories, namely keeping entangled pairs until the binding phase. The~practical problem of long-term stable quantum memories can be overcome by a simple modification: instead of {EPR} pairs, Trent sends (over an authenticated quantum channel) two ordered sets of photons in pure BB84 states, one for Alice and the other for Bob. As~soon as the clients receive their particles, they measure on them the honest observables according to $h(M)$. Thus, all~the information kept by the parties is classical and can be used to check the agents' behavior and honesty, similarly as in the protocol we proposed. This approach goes along the lines of reducing the security of BB84 to the {E91} entangled QKD protocol. The~details of such a reduction are a matter of a separate~study.

Contract signing is an important and wide-spread cryptographic protocol, as~performing transactions over the Internet is an important part of today's society. Nevertheless, all such current classical implementations fully rely on the use of a third trusted party. Moreover, if~buying directly from, say, Amazon, the~trust is handed over to a signing party, Amazon itself. Within~classical approaches, there are several service providers and applications that can mediate the contract signing process, such as EverSign, HelloSign, and~DocuSign.

On the other hand, if~the parties do not want to rely fully on a third party to exchange the signature between them, classical solutions, such as~\cite{even:83,gold:84,even:85,asokan:97,asokan:98,ben:90,rabin:83}, mentioned in the Introduction, require the gradual exchange of signed parts of the message. Therefore, classical digital contract signing is a very demanding application from the point of view of communication and computation: the~exchange process has a significant number of rounds, and~each exchanged message has to be digitally signed using (computationally-demanding) public-key cryptographic systems and infrastructures. We~refer the reader to the following survey~\cite{rak:14} for a more detailed description of these protocols. As~a consequence, such classical contract signing applications are, due to their current inefficiency, to~the best knowledge of the authors, not present on the market. We note that this is precisely the reason classical solutions to secure multiparty computation privacy protocols are not widely available as well. Potentially fast exchange of (single or coherent) photon pulses might be one of the main advantages of quantum solutions to the mentioned cryptographic problems. Note that by now, it was the higher security levels that promoted quantum over classical cryptography. This novel feature, providing practical schemes not even available classically, might potentially be shown to be~significant.

\section{Conclusions} \label{sec:conclusion}

We presented a quantum protocol for signing contracts. We showed that, under~the realistic assumptions of noise and measurement errors, the~protocol was fair, and~consequently optimistic as well. In~particular, the~maximal probability to cheat can be made arbitrarily low, as~it scaled as $1/\sqrt{N}$, where $4N$ was the total number of rounds of the protocol. We also showed that our protocol was robust against noise. Indeed, even in a dishonest noisy case, the~probability to cheat could be as low as $0.022$, for~$4N=6000$, the~noise parameter $\kappa = 0.05$, and~$\alpha \in [0.7,0.99]$. Moreover, already for $4N=400$, the~probability to cheat was below $10\%$ (in the honest scenarios, this was so for $4N=300$). Given the noise level $\kappa$, finding optimal parameters $N$, the~range of $\alpha$ and $p(\alpha)$, was straightforward using the ``brute force'' techniques presented in the current paper (running computer codes and simulations). Full analytical study of such optimization problem is a matter of future~study.

Unlike the classical counterparts, the~protocol was based on the laws of physics, and~the clients did not need to exchange a huge number of signed authenticated messages during the actual contract signing process (the exchange phase). Thus, the~protocol was abuse-free: a client cannot prove to be involved in the actual act of signing the contract. In~contrast, in~the classical counterparts, having the signed and authenticated messages received during the exchange phase, a~client could show them to other interested clients to negotiate better terms of a financial transaction. In~other words, classically, a~dishonest client can abuse the signing process by falsely presenting his/her interest in the deal, while~actually using the protocol to achieve a different goal(s). Unlike generic quantum security protocols (say, quantum key distribution), preforming quantum measurements different from those prescribed by the protocol cannot help a cheating client (see Appendix~\ref{sec:dishonest_noisy} for details). In~the current proposal, each client can independently obtain the signed contract, without~the other client being present, which was not possible in~\cite{paun:2011}. Thus, the~probability to cheat was assigned to a real event (whereas in~\cite{paun:2011}, it was just a formal figure of merit). Unlike the classical counterparts (and the previous quantum proposal~\cite{paun:2011}), when first contacting Trent, the~clients did not need to agree upon a definitive contract. Moreover, Trent never learned the actual content of the protocol, as~the clients provided Trent its hash value, $h^\ast=h(M)$, given by publicly-known hash function $h$.

\begin{acknowledgments}
The authors acknowledge the support of SQIG -- Security and Quantum Information Group, the Instituto de Telecomunica\c{c}\~oes (IT) Research Unit, ref. UID/EEA/50008/2019, funded by Funda\c{c}\~ao para a Ci\^encia e Tecnologia (FCT), and the FCT projects Confident PTDC/EEI-CTP/4503/2014, QuantumMining POCI-01-0145-FEDER-031826, Predict PTDC/CCI-CIF/29877/2017, supported by the European Regional Development Fund (FEDER), through the Competitiveness and Internationalization Operational Programme (COMPETE 2020), and by the Regional Operational Program of Lisbon. A.S. acknowledges funds granted to LaSIGE Research Unit, ref. UID/CEC/00408/2013. P.Y. acknowledges the support of DP-PMI and FCT (Portugal) through the scholarship PD/BD/113648/2015. The authors acknowledge Marie-Christine R$\ddot{\mbox{o}}$hsner for pointing out an attack on the protocol, which lead to improvement of the work. N.P. acknowledges Aleksandra Dimi\'c for fruitful discussions.
\end{acknowledgments}

\onecolumngrid

\appendix
\section*{APPENDIX}
In this Appendix, we analyze protocol's security in case the communication is interrupted at step $m$. To proceed to the Binding phase, the two clients measure their respective Honest observables ($\hat{H}_{\texttt{A}_i}$ and $\hat{H}_{\texttt{B}_i}$) on the rest of their qubits, thus increasing their probability to pass Trent's test 10a on qubits shared with him. For the qubits shared between Trent and the other client, used in test 10b, the best they can do is to guess their outcomes. Note that this strategy is optimal for both the case of honest clients, studied in Appendix~\ref{sec:soundness}, as well as for a dishonest client, discussed in Appendix~\ref{sec:dishonest_noisy}.

Before discussing protocol's correctness for the case of honest, and security for dishonest clients, we present the overall set-up: the ``configuration'' $\mathcal L$ and the details of the structure of the figure of merit, the expected probability to cheat for Bob and Alice. We assume the scenario in which the communication is interrupted upon both clients received $m$ measurement outcomes from the other. The biased case of one client having $m$, while the other $m-1$, results is described analogously. A simplified protocol description is given in Figure~\ref{fig:setup}.

\begin{figure}
\centering
\includegraphics[width=0.56\linewidth]{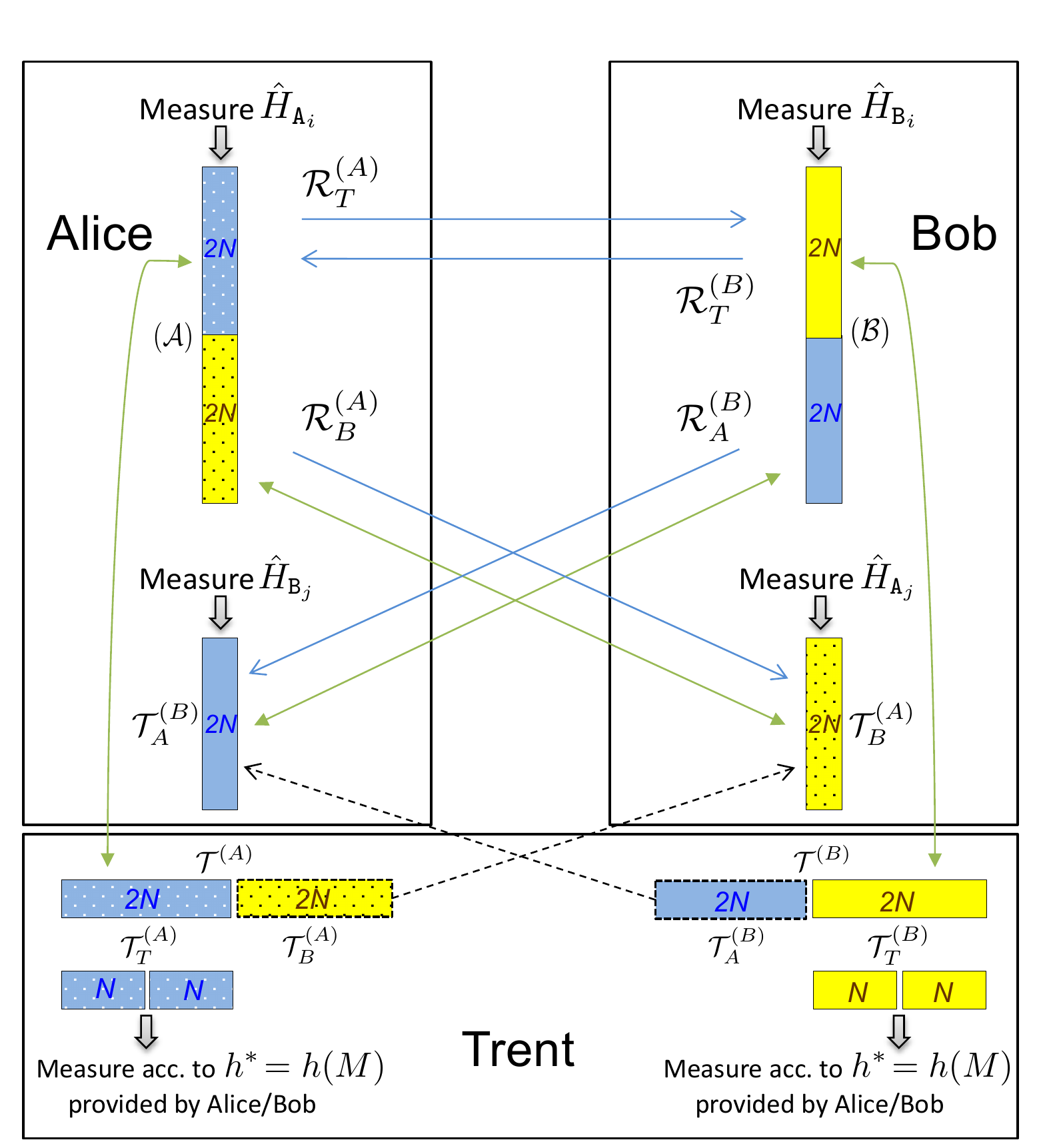}
\caption{\footnotesize{Stages of the protocol. The green arrows represent entanglement between the corresponding qubits. The blue arrows represent the transfer of measurement outcomes. The dashed arrows represent the transfer of qubits $\mathcal{T}_A^{(B)}$ and $\mathcal{T}_B^{(A)}$ from Trent to Alice and Bob, respectively. The big boxes represent shielded private laboratories of Alice, Bob and Trent.}}
\label{fig:setup}
\end{figure}

\section{Probabilistic fairness for honest clients} \label{sec:soundness}
\unskip
In this scenario, we assume that both Alice and Bob are honest until step $m$ after which the communication is interrupted. Without loss of generality, we assume that Alice was the first one to start the information exchange. To analyze the fairness of the protocol, we consider the case when the communication is interrupted after Alice has sent $m$ outcomes to Bob but received only $m-1$ outcomes in return, so Bob has a slight advantage over Alice. The case where both Alice and Bob have $m$ outcomes each is perfectly symmetric (that is, both Bob and Alice have the same expected probability to bind). In this Section, we will analyze the ideal case of a noiseless channel, as well as the protocol's robustness in the presence of noise.

\subsection{Noiseless channel} \label{sec:honest_noiseless}

In the absence of any noise, both honest Alice and Bob obtain perfect correlations on their respective qubits shared (entangled) with Trent and among themselves. Hence, both clients will, with certainty, pass each other verification tests, $P_{ABS}(m)=P_{BAS}(m)=1$, and thus $P(m) = P_{ABS}(m)\cdot P_{BAS}(m) = 1$ (i.e., the reason for communication interruption is the failure of the network). Then, the probability to cheat, say, for Bob~\eqref{cheat}, is given by 
\begin{equation} \label{eq:honest_cheat}
P_{ch}^\bob (m;\alpha) = P^\bob_{bind}(m;\alpha)\left[1-P^\alice_{bind}(m;\alpha)\right]. 
\end{equation}

Moreover, the probabilities for the clients to pass Trent's test on their own qubits, $P_{ATH}(m;\alpha)$ and $P_{BTH}(m;\alpha)$, are also $1$. Thus, Bob's probability to bind~\eqref{eq:Pm} is also simplified, giving
\begin{equation}
P^\bob_{bind}(m;\alpha) = P_{BTH}(m;\alpha) P_{BTA}(m;\alpha) = P_{BTA}(m;\alpha)
\end{equation}
and analogously for Alice.

Out of the $m$ correct results that Bob received from Alice (note that we assume ideal noiseless scenario), $N-\ell_{T2}^{(A)}$ are relevant when presenting to Trent to prove Alice's commitment. With $\alpha$ chosen randomly by Trent, let $\lfloor \alpha N\rfloor$ be the number of correct results corresponding to $N$ qubits from $\mathcal{T}_{T2}^{(A)}$ and $\mathcal{T}_{T2}^{(B)}$ each, that Trent needs to receive from Bob in order to bind the contract for him. For the case $\lfloor \alpha N\rfloor \leq N - \ell_{T2}^{(A)}$, Bob already has more than the required number of correct results, hence his probability to bind the contract will be $1$. When $\lfloor \alpha N\rfloor > N - \ell_{T2}^{(A)}$, Bob must correctly guess at least $n_{c}=\lfloor \alpha N\rfloor -\left(N-\ell_{T2}^{(A)}\right)$ out of the remaining $\ell_{T2}^{(A)}$ results to convince Trent to bind the contract for him. Hence, his probability to bind the contract is given by
\begin{equation}\label{eq:1}
P_{bind}^{\bob}(m;\alpha|\ell_{T2}^{(A)})=P_{BTA} (m;\alpha|\ell_{T2}^{(A)})=
\begin{cases}
1 & \:\:\:\:\:\:\: \textrm{if } \lfloor \alpha N\rfloor \leq N-\ell_{T2}^{(A)} \vspace*{1mm}\\
\ds 2^{-\ell_{T2}^{(A)}} \sum\limits_{n=n_{c}}^{\ell_{T2}^{(A)}} \dbinom{\ell_{T2}^{(A)}}{n} & \:\:\:\:\:\:\: \textrm{if } \lfloor \alpha N\rfloor>N-\ell_{T2}^{(A)}.
\end{cases}
\end{equation}
Here, $\dbinom{\ell_{T2}^{(A)}}{n}$ gives the number of possible choices for $n$ out of $\ell_{T2}^{(A)}$, and $2^{-\ell_{T2}^{(A)}}=2^{-n} \ 2^{-(\ell_{T2}^{(A)} -n)}$ gives the probability of guessing correctly exactly $n$ results (incorrectly on the rest $\ell_{T2}^{(A)} -n$ results).

Averaging over $\ell_{T2}^{(A)}$ gives the binding probability for Bob, as a function of measurement outcomes obtained from Alice, $m$ (the round of communication interruption), and the parameter $\alpha$
\begin{equation}
\label{eq:bobbindexact}
P_{bind}^{\bob}(m;\alpha)=\sum_{\ell_{T2}^{(A)}} q\left(\ell_{T2}^{(A)}\right) P_{bind}^{\bob}\left(m;\alpha|\ell_{T2}^{(A)}\right) ,
\end{equation}
with the probability distribution for $\ell_{T2}^{(A)}$ given by
\begin{equation}
\label{eq:B_ell}
q\left(\ell_{T2}^{(A)}\right)=\frac{\dbinom{m}{N-\ell_{T2}^{(A)}} \dbinom{4N-m}{\ell_{T2}^{(A)}}}{\dbinom{4N}{N}}.
\end{equation}
To verify the last expression, note that at the $m$-th step: (i) there are $\dbinom{m}{N-\ell_{T2}^{(A)}}$ ways of choosing the $N-\ell_{T2}^{(A)}$ results already obtained from Alice; (ii) $\dbinom{4N-m}{\ell_{T2}^{(A)}}$ ways of choosing the $\ell_{T2}^{(A)}$ results from the remaining $4N-m$ to guess; (iii) there are $\dbinom{4N}{N}$ ways of choosing the $N$ results, relevant for binding the contract for Bob, from the total $4N$ results.

Note that in general $\ell_{T2}^{(A)}$ takes values from $0$ to $N$, but for the step $m$ of communication interruption, $\ell_{T2}^{(A)}$ is constrained to the following values 
\begin{equation}
\label{eq:ranges}
\begin{array}{rcll}
  \ell_{T2}^{(A)} & = & N-m\:, \dots ,N   & \:\:\:\:\:\:\:\: \mbox{for } m < N, \vspace{1.7mm}\\
  \ell_{T2}^{(A)} & = & 0\:, \dots ,N  &  \:\:\:\:\:\:\:\: \mbox{for } N \leq m < 3N, \vspace{1.7mm}\\
  \ell_{T2}^{(A)} & = & 0\:, \dots ,4N-m  &  \:\:\:\:\:\:\:\: \mbox{for } 3N \leq m.
\end{array}
\end{equation}
The above ranges define our summation $\sum_{\ell_{T2}^{(A)}}$. When $\ell_{T2}^{(A)}$ is out of the first or the third range, we have $m < N - \ell_{T2}^{(A)}$, and $4N-m < \ell_{T2}^{(A)}$, respectively. Since in those cases the binomials $\dbinom{m}{N-\ell_{T2}^{(A)}}$ and $\dbinom{4N-m}{\ell_{T2}^{(A)}}$ are by definition equal to zero, we can always take the $\ell_{T2}^{(A)} = 0\:, \dots ,N$ range.

\begin{figure}\centering          
\begin{tabular}{cc}
\includegraphics[width=6.5cm]{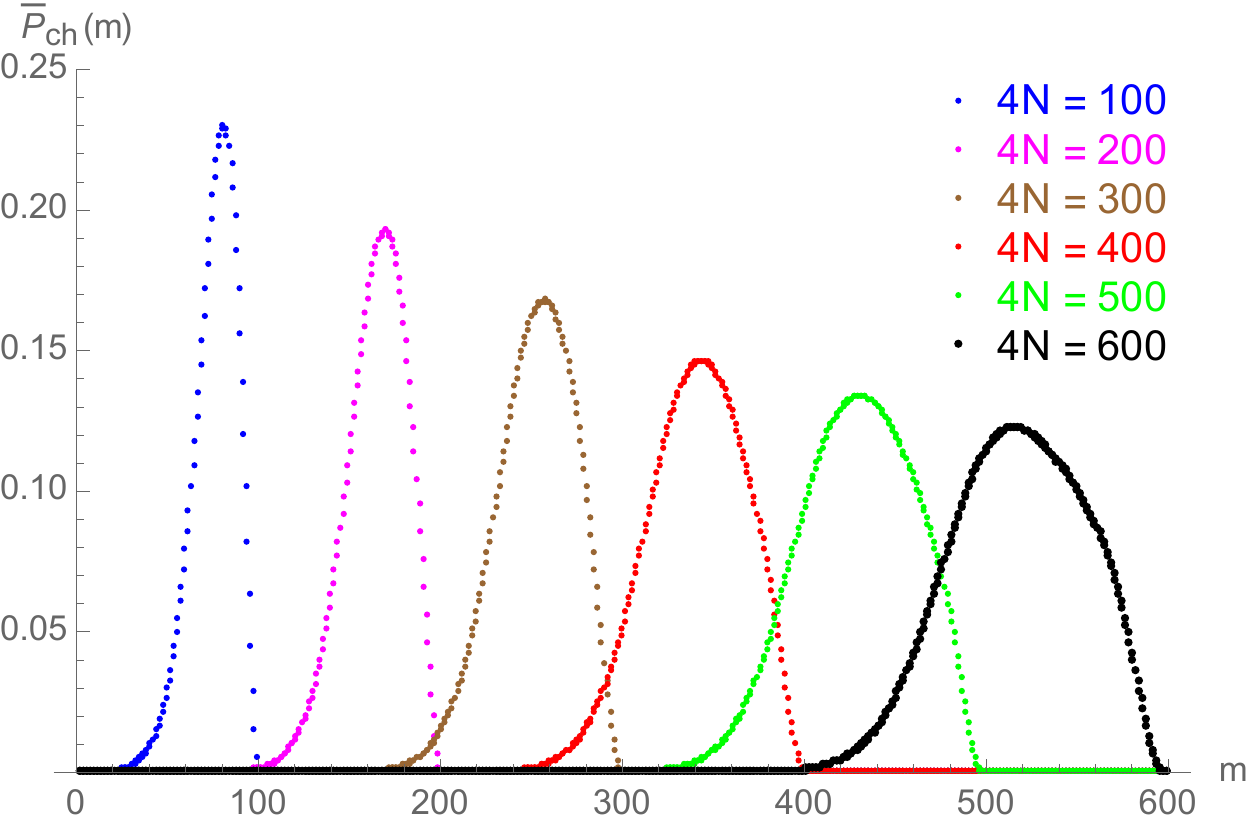}   \ \ \ \      &
\includegraphics[width=6.5cm]{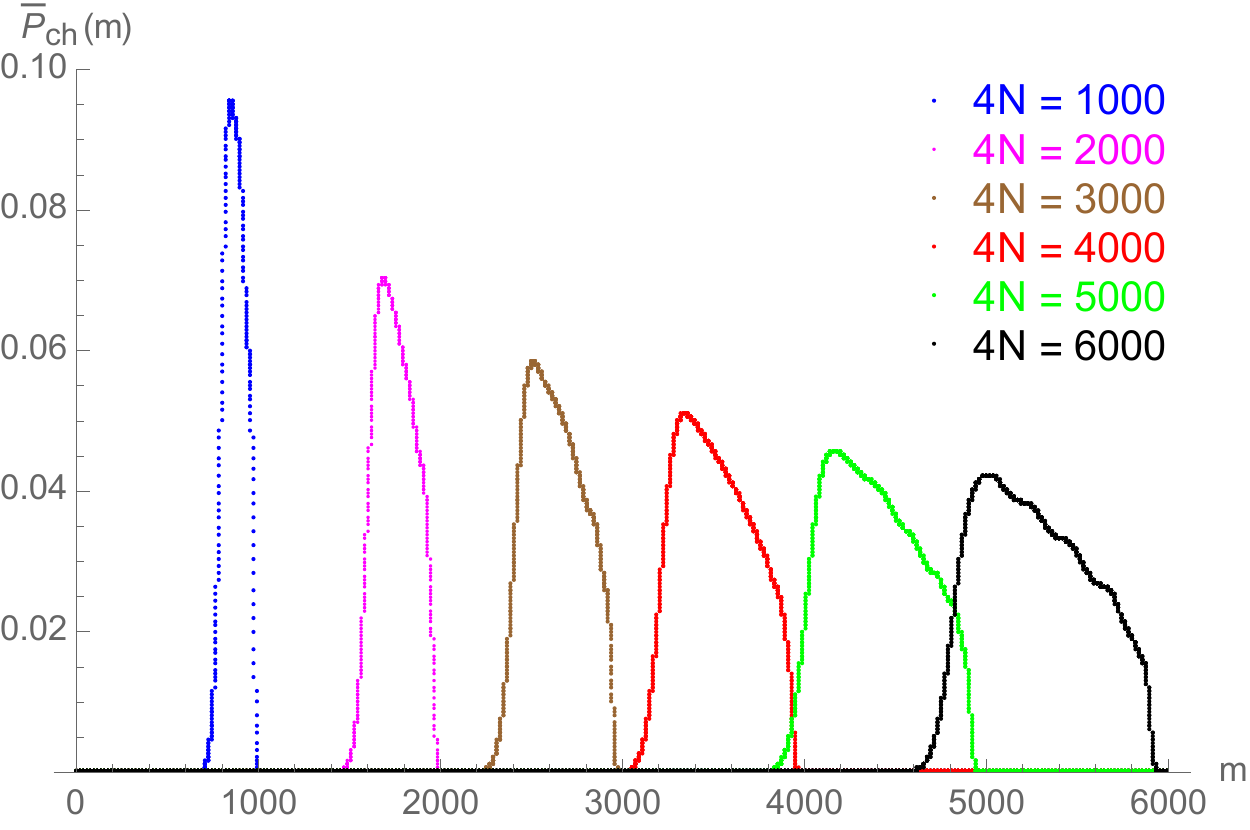} \tabularnewline
\end{tabular}
\caption{\footnotesize{Honest scenario (noiseless channel): the expected probability to cheat, $\bar{P}_{ch}(m)$, is plotted against the step of communication interruption $m$ (from $1$ to $4N$), for $\alpha$ chosen uniformly on the interval $[0.9,0.99]$.}}
\label{fig:honest_noiseless2D}
\end{figure}

By replacing $m$ with $m-1$ in equation~\eqref{eq:B_ell}, one obtains Alice's probability to bind the contract $P_{bind}^{\alice}(m;\alpha)$ (note that, having one less measurement results than Bob, she therefore has a small disadvantage).

Bob's probability to cheat~\eqref{eq:honest_cheat}, when the communication is interrupted at step $m$, and for a fixed $\alpha$, is given by $P_{ch}^\bob (m;\alpha)=P^\bob_{bind}(m;\alpha)\left[1-P^\alice_{bind}(m;\alpha)\right].$ Using~\eqref{eq:bobbindexact} and its counterpart for Alice one can evaluate the expected probability to cheat $\bar{P}^\bob_{ch}(m) = \int p(\alpha)\: P_{ch}^\bob (m;\alpha)\: \mbox d \alpha$ for every $m$. In Figure~\ref{fig:honest_noiseless2D}, we plot the expected probability to cheat, $\bar{P}^\bob_{ch}(m)$, against the step of communication interruption $m$ (running from $1$ to $4N$), for the simplest case of the uniform distribution $p(\alpha)$ on the interval $[0.9,0.99]$. Since Alice starts first, and Bob is thus privileged, we have that $\bar{P}^\alice_{ch}(m) \leq \bar{P}^\bob_{ch}(m) = \bar{P}_{ch}(m)$. The maximal value shows the behavior $\max_m \bar{P}_{ch}(m) \propto N^{-1/2}$, as presented in Figures~\ref{fig:plots}~(a) and~\ref{fig:plots}~(d), from the main text.

In addition to the quantitative results for up to $4N = 6000$, below we present an analytic proof of the asymptotic behavior for the maximal expected probability to cheat, showing that as $N \rightarrow \infty$ we have $P_{ch}(m;\alpha) \propto N^{-1/2}$. The proof for the case of a noisy channel follows analogously (see~\cite{hana:17} for the analysis for the contract signing presented in~\cite{paun:2011}).

\begin{theorem}
\label{theorem} 
In the honest noiseless case, for uniformly chosen $\alpha\in (1/2,1)$, we have that
$$P_{ch}(m;\alpha) \in \mathcal O (N^{-1/2}).$$
\end{theorem}

\begin{proof}
We are going to show that for all $1<m<4N$ and $\alpha\in (1/2,1)$ we have that $P_{ch}(m;\alpha) \in \mathcal O (N^{-1/2})$ and therefore, the statement follows straightforwardly.

Recall that a Binomial distribution $B(k,p)$, with $k$ sufficiently large and $p$ bounded away from $0$ and $1$ (that is, does not tend to 0 or 1 as $N$ grows to infinity), can be approximated by a Normal distribution with mean, $kp$, and variance, $kp(1-p)$ as $\mathcal N (kp,kp(1-p))$. In our case (with $k=\ell_{T2}^{(A)})$, for $\ell_{T2}^{(A)}$ sufficiently large and $p=1/2$ (guessing probability of a binay bit), the Binomial distribution, $B\left(\ell_{T2}^{(A)}, 1/2 \right)$ can be approximated by $\mathcal N \left(\ell_{T2}^{(A)}/2\ ,\ \ell_{T2}^{(A)}/4\right)$. Hence, approximating $P_{bind}^{\bob}\left(m;\alpha|\ell_{T2}^{(A)}\right)$ for the case $\lfloor\alpha N\rfloor>N-\ell_{T2}^{(A)}$, we have
\begin{equation}
\label{eq:bobbindapprox0}
	\begin{array}{rcl}
		P_{bind}^{\bob}(m;\alpha|\ell_{T2}^{(A)}) &=& \ds 2^{-\ell_{T2}^{(A)}} \sum\limits_{n=n_c}^{\ell_{T2}^{(A)}} \dbinom{\ell_{T2}^{(A)}}{n} \\[3mm] 
		&\approx & \ds \int_{n_c}^{\infty} \frac{1}{\sqrt{2\pi \sigma^2}} e^{-\frac{(n-\mu_B)^2}{2\sigma_B^2}} \mbox{d} n = \frac 1 2 \mbox{erfc} \left[ \sqrt{\frac{2}{\ell_{T2}^{(A)}}}\left(n_c - \frac{\ell_{T2}^{(A)}}{2} \right) \right] \, ,	
	\end{array}
\end{equation}
with $n_c = \lfloor \alpha N\rfloor -(N-\ell_{T2}^{(A)})$, and Binomial mean and variance given by $\mu_B=\ell_{T2}^{(A)}/2$ and $\sigma_B^2=\ell_{T2}^{(A)}/4$, respectively. Therefore
\begin{equation}
\label{eq:bobbindapprox}
P^\bob_{bind}(m;\alpha|\ell_{T2}^{(A)})\approx
\begin{cases}
1 & \:\:\:\:\:\:\:\:\:\:\:\:\:\:\:\:\:\:\:\: \textrm{if } \lfloor\alpha N\rfloor\leq N-\ell_{T2}^{(A)} \, , \vspace*{2mm}\\
\ds \frac 1 2 \ \mbox{erfc}\left[ \sqrt{\frac{2}{\ell_{T2}^{(A)}}}\left( N(\alpha -1) + \frac{\ell_{T2}^{(A)}}{2} \right) \right] & \:\:\:\:\:\:\:\:\:\:\:\:\:\:\:\:\:\:\:\: \textrm{if } \lfloor\alpha N\rfloor>N-\ell_{T2}^{(A)}\, ,
\end{cases}
\end{equation}
where $\mbox{erfc}(\beta)$ is the complementary error function, defined as
\begin{equation}
\mbox{erfc}(\beta)= \frac{2}{\sqrt{\pi}} \int_{\beta}^{\infty} e^{-{\gamma}^2} \mbox{d} \gamma\, .
\end{equation}
For the case $\lfloor\alpha N\rfloor>N-\ell_{T2}^{(A)}$, we expand $P^\bob_{bind}(m;\alpha|\ell_{T2}^{(A)})$ around $\ell_{T2}^{(A)}= 2N(1-\alpha)$ where it has the value $1/2$, to obtain
\begin{eqnarray}\label{eq:expand}
\frac 1 2 \ \mbox{erfc}\left[ \sqrt{\frac{2}{\ell_{T2}^{(A)}}}\left( N(\alpha -1) + \frac{\ell_{T2}^{(A)}}{2} \right) \right]\!\!\!\!\!\!\!\!\!\!&&=\frac 1 2 + c_1 (\ell_{T2}^{(A)} -2N(1-\alpha)) + c_2 (\ell_{T2}^{(A)} -2N(1-\alpha))^2 \nonumber\\
&& \ \ \ \ \ \ + \ \! \mathcal{O}((\ell_{T2}^{(A)} -2N(1-\alpha))^3 ).
\end{eqnarray}
Let $P_1 (m;\alpha|\ell_{T2}^{(A)})=\frac 1 2 + c_1 (\ell_{T2}^{(A)} -2N(1-\alpha)) + c_2 (\ell_{T2}^{(A)} -2N(1-\alpha))^2$. We also obtain the limits for $\ell_{T2}^{(A)}$ where this function becomes $1$ and $0$, denoted by $\ell_{inf}$ and $\ell_{sup}$, respectively.

We know that
\begin{equation}\label{eq:avgBob}
P_{bind}^{\bob}(m;\alpha)=\sum_{\ell_{T2}^{(A)}} q\left(\ell_{T2}^{(A)}\right) P_{bind}^{\bob}\left(m;\alpha|\ell_{T2}^{(A)}\right)\, ,
\end{equation}
with the probability distribution for $\ell_{T2}^{(A)}$ given by
\begin{equation} \label{eq:ellBob}
q\left(\ell_{T2}^{(A)}\right)=\frac{\dbinom{m}{N-\ell_{T2}^{(A)}} \dbinom{4N-m}{\ell_{T2}^{(A)}}}{\dbinom{4N}{N}}\, .
\end{equation}
By Feller's result~\cite{feller:68} on the approximation of hypergeometric distribution, one can approximate~\eqref{eq:ellBob} for $N\rightarrow\infty$, $\frac{N}{4N}\rightarrow t\in (0,1)$ to a Normal distribution as
\begin{equation}\label{eq:fellBob}
q(\ell_{T2}^{(A)})\sim \frac{e^{-(\ell_{T2}^{(A)}-\mu_{\ell})^2/2\sigma_{\ell}^2(1-t)}}{\sqrt{2\pi \sigma_{\ell}^2(1-t)}}=P_2 (\ell_{T2}^{(A)})\, ,
\end{equation}
with mean, $\mu_\ell=N-m/4$, variance, $\sigma_{\ell}^2=N(\frac{m}{4N})(\frac{4N-m}{4N})$, and $t=1/4$. Hence,~\eqref{eq:avgBob} can be upper bounded (shifting the limit $\ell_{T2}^{(A)}=N(1-\alpha)$ to $\ell_{T2}^{(A)}=\ell_{inf}$ and using the approximation~\eqref{eq:expand} from $\ell_{inf}$ to $\ell_{sup}$) by
\begin{eqnarray}
\nonumber P_{bind}^{\bob}(m;\alpha)\!\! && \!\!\!\!\!\!\!\! \leq 
\int_{0}^{\ell_{inf}} 1\!\cdot\! P_2 (\ell_{T2}^{(A)})\; \mbox d \ell_{T2}^{(A)} +
\underbrace{\int_{\ell_{inf}}^{\ell_{sup}} P_2 (\ell_{T2}^{(A)})\! \cdot\! \left[P_1 (m;\alpha|\ell_{T2}^{(A)}) + k\left(\ell_{T2}^{(A)} -2N(1-\alpha)\right)^3 \right] \mbox d \ell_{T2}^{(A)}}_{B_1}\\
& & \;\;\;\;\;\;\;\;\;\;\;\;\;\;\;\;\;\;\;\;\;\;\;\;\;\;\;\;\;\;\;\;\;\;\;\;\;\;\;\;\;\;\;\;\; 
+ \underbrace{\int_{\ell_{sup}}^{+\infty} P_2 (\ell_{T2}^{(A)}) \cdot P_{bind}^{\bob} (m;\alpha|\ell_{T2}^{(A)})\; \mbox d \ell_{T2}^{(A)}}_{B_2}\, \\
&&  \!\!\!\!\!\!\!\! \leq \int_{0}^{\ell_{inf}} P_2 (\ell_{T2}^{(A)})\; \mbox d \ell_{T2}^{(A)} + \mathcal O (N^{-1/2}),\label{eq:boundBob}
\end{eqnarray}
where the $\mathcal O (N^{-1/2})$ in Eq.~\eqref{eq:boundBob} was obtained by computing limits for $B_1$ and $B_2$ for $N\rightarrow \infty$ with the Mathematica software package. To compute the limit for $B_1$, from~\cite{B:2000} we used the fact that by change of variable $N=1/x$, if $\lim_{x\rightarrow 0^{+}} f(x)=0$ then $\lim_{N\rightarrow \infty} f(N)=0$ and obtained the desired bound. For computing $B_2$, we upper bounded $P_{bind}^{\bob} (m;\alpha|\ell_{T2}^{(A)})$ by $\frac 1 2 \mbox{erfc}[\lambda N]$ for some constant $\lambda>0$ and used the upper bound for complementary error function, which also led to the $O (N^{-1/2})$ bound. 

To obtain the expression for Alice's probability to bind, $P^\alice_{bind}(m;\alpha)$, we have
\begin{equation} \label{eq:ellAlice}
q\left(\ell_{T1}^{(B)}\right)=\frac{\dbinom{m-1}{N-\ell_{T1}^{(B)}} \dbinom{4N-m+1}{\ell_{T1}^{(B)}}}{\dbinom{4N}{N}}\, .
\end{equation}
Making the same approximations, we get $P_2\left(\ell_{T1}^{(B)}\right)$ similar to~\eqref{eq:fellBob}, with $m$ replaced by $m-1$. Hence
\begin{eqnarray}
\nonumber P_{bind}^{\alice}(m;\alpha)&\geq & 
\int_{0}^{\ell_{inf}} P_2 (\ell_{T1}^{(B)})\; \mbox d \ell_{T1}^{(B)} - 
\underbrace{\int_{0}^{\ell_{inf}} (1-P_{bind}^{\alice}(m;\alpha|\ell_{T1}))P_2 (\ell_{T1}^{(B)})}_{A_1}\; \mbox d \ell_{T1}^{(B)} \\ & &
 \;\;\;\;\;\;\;\;\;\;\;\;\;\;\;\;\;\;\;\;\;\;\;\;\;\;\;\;\;\;\;\;\;\;\;\;\;\;\;\;\;\;\; +
 \underbrace{\int_{\ell_{inf}}^{+\infty} P_2 (\ell_{T1}^{(B)}) \;P_{bind}^{\alice}(m;\alpha|\ell_{T1}^{(B)})\; \mbox d \ell_{T1}^{(B)}}_{A_2}\\
& \geq & \int_{0}^{\ell_{inf}} P_2 (\ell_{T1}^{(B)})\; \mbox d \ell_{T1}^{(B)}-\mathcal O (N^{-1/2}),\label{eq:boundAlice}
\end{eqnarray}
where the bounds in Eq.~\eqref{eq:boundAlice} was obtained with the help of the Mathematica software package. Since we are computing the lower bound, one can drop $A_2$ and it suffices to show that the integral of $A_1$ computed between $N(1-\alpha)$ (as $P_{bind}^{\alice}(m;\alpha|\ell_{T1})$ is $1$ upto $N(1-\alpha)$) and $\ell_{inf}$ vanishes when $N\rightarrow \infty$. Since $(1-P_{bind}^{\alice}(m;\alpha|\ell_{T1}))<1$, we can drop this and compute the limit of the integral of $P_2 (\ell_{T1}^{(B)})$, which also provides the bound.

If the communication was interrupted at step $m$, for all choices of $\alpha \in (1/2,1)$ by Trent, we want to compute, say for Bob, the probability to cheat, given by
\begin{eqnarray}
P_{ch}^\bob (m;\alpha)&= &P^\bob_{bind}(m;\alpha)\left[1-P^\alice_{bind}(m;\alpha)\right]\\
   &\leq& \left[ \int_{0}^{\ell_{inf}}  P_2 (\ell_{T2}^{(A)})\; \mbox d \ell_{T2}^{(A)} + O(N^{-{\frac{1}{2}}})\right] \left[1- \int_{0}^{\ell_{inf}} P_2 (\ell_{T1}^{(B)})\; \mbox d \ell_{T1}^{(B)}+\mathcal O (N^{-1/2})\right]\\
   &=& \mathcal O (N^{-1/2}),\label{eq:finalbound}
\end{eqnarray}
where the $\mathcal O (N^{-1/2})$ in Eq.~\eqref{eq:finalbound} was again obtained by computing a limit with the Mathematica software package.
\end{proof}

\subsection{Noisy channel} \label{sec:honest_noisy}

For the case of a noisy channel, a binomial test could be used by both parties for the permitted number of wrong results from the other party. Consider white noise in the channel that decreases the degree of correlation between the honest clients' results. The depolarizing channel, modeling the effects of white noise on a two-qubit mixed state $\rho$, is given by ($I$ is the identity matrix for dimension 4, and $\kappa\in[0,1]$):
\begin{equation}
\mathcal{E}_{d}(\rho)=(1-\kappa)\rho + \kappa \frac{I}{4}.
\end{equation}

Hence, both Alice and Bob receive some inevitable number of incorrect results, due to the noise considered, in spite of both of them measuring their respective Honest observables on their own qubits. Let us denote by $p(xy)$ the probability of Alice and Bob obtaining the results $x$ and $y$, respectively, corresponding to measurements of $\mathcal{E}_{d}(\rho)$ in either computational or diagonal basis, for the case of the entangled two-qubit state $\rho = |\psi^+\rangle\langle\psi^+|$. Given both Alice and Bob measure their respective correct observables, $\hat{H}_{\texttt{A}_i}$ and $\hat{H}_{\texttt{B}_i}$, we have the following probabilities of obtaining different results
\begin{equation}
\begin{array}{rl}
p(00) = & (2 - \kappa)/4, \vspace{1.5mm}\\ 
p(01) = & \kappa/4, \vspace{1.5mm} \\
p(10) = & \kappa/4, \vspace{1.5mm} \\
p(11) = & (2 - \kappa)/4.
\end{array}
\end{equation}
For a given probability of the favorable event (in our case, $p(00)+p(11)$, of obtaining consistent results), and the total number of such events (say, $m_r$), the probability of obtaining exactly $r$ correct results is given by the binomial distribution
\begin{equation}
\mathcal P \left(r | |\psi^{+}\rangle ; m_r,\kappa \right)= \binom{m_r}{r} \Big(p(00)+p(11)\Big)^{r} \Big(p(01)+p(10)\Big)^{m_r -r},
\end{equation}
with mean and variance given by
\begin{equation}
\label{eq:meanvar}
\begin{array}{rl}
\mu =& \ds m_r \Big(p(00)+p(11)\Big) = m_r \left(1-\frac{\kappa}{2}\right),\vspace{1.5mm}\\
\sigma^2 =& \ds m_r \Big(p(00)+p(11)\Big) \Big(1-p(00)-p(11)\Big) = m_r \left(1-\frac{\kappa}{2}\right) \frac{\kappa}{2}.
\end{array}
\end{equation}
Suppose that each party applies a 3-sigma ``acceptance criterion'', then a client will continues as long as he/she has at least $\mu-3\sigma$ consistent results from the other.

From~\eqref{eq:Pm}, Bob's probability to bind the contract is $P_{bind}^\bob (m;\alpha)=\langle P_{BTH}(m;\alpha|\ell_{T2}^{(B)}) \rangle_{\ell^{(B)}_{T2}} \cdot \langle P_{BTA}(m;\alpha|\ell_{T2}^{(A)})  \rangle_{\ell^{(A)}_{T2}}$. While Bob's probability to pass the test on his qubits, $P_{BTH}$, does not depend on step $m$ {\em nor on} $\ell_{T2}^{(B)}$ (he measures all of his qubits), the probability to pass the test on Alice's qubits, $P_{BTA}$, depends on $m$, {\em as well as} on $\ell_{T2}^{(A)}$ (for simplicity, we omit the implicit dependence on the noise parameter $\kappa$). Thus, we have
\begin{equation}\label{eq:BTHBTA}
P_{bind}^\bob (m; \alpha) = P_{BTH}(\alpha) \cdot \langle P_{BTA}(m;\alpha|\ell_{T2}^{(A)})  \rangle_{\ell^{(A)}_{T2}},
\end{equation}
with
\begin{eqnarray}\label{eq:BTH}
\ds P_{BTH} (\alpha) \!\!& = &\!\! 
\underbrace{
\sum_{s = \lfloor \alpha N \rfloor}^N \binom{N}{s} P_{=}^s  P_{\neq}^{N - s},
}_{\substack{\textnormal{probability of obtaining at least} \\ \textnormal{$\lfloor \alpha N \rfloor$ correct results from $N$ by measuring} \\ \textnormal{the correct observable ($\hat{H}_{\texttt{B}_i}$) on all his qubits}}}
\\
\label{eq:BTA}
\ds P_{BTA} (m;\alpha) \!\!& = &\!\! \sum_{\ell_{T2}^{(A)}} q\left(\ell_{T2}^{(A)}\right) \sum_{ t = 0 }^{N- \ell_{T2}^{(A)}} 
	\underbrace{
	\binom{N- \ell_{T2}^{(A)} }{t} P_{=}^{t} P_{\neq}^{(N - \ell_{T2}^{(A)}) - t}
	}_{\substack{\textnormal{probability of obtaining $t$ correct} \\ \textnormal{results from the $N- \ell_{T2}^{(A)}$ results} \\ \textnormal{received upto step $m$}}}\;\; \times
	\underbrace{
	2^{-\ell_{T2}^{(A)}} \sum_{u = \lfloor \alpha N \rfloor - t}^{\ell_{T2}^{(A)}} \binom{\ell_{T2}^{(A)}}{u} 
	}_{\substack{\textnormal{probability of guessing} \\ \textnormal{at least $\lfloor \alpha N \rfloor - t$ correct results} \\ \textnormal{from the $\ell_{T2}^{(A)}$ results to guess}}},
\end{eqnarray}
where $P_{=} = p(00) + p(11) $ and $P_{\neq} =p(10) + p(01)$. Analogously, we define $P_{bind}^\alice (m; \alpha) = P_{ATH}(\alpha) \cdot \langle P_{ATB}(m;\alpha|\ell_{T1}^{(B)})  \rangle_{\ell^{(B)}_{T1}}$ for Alice.

To obtain Bob's probability to cheat~\eqref{cheat}, $P_{ch}^\bob (m;\alpha) = P(m)P^\bob_{bind}(m;\alpha)\left[1-P^\alice_{bind}(m;\alpha)\right]$, we now estimate $P(m) = P_{ABS}(m) \cdot P_{BAS}(m)$. $P_{BAS}(m)$, the probability that Bob passes Alice's tests on their Shared qubits (from the $m-1$ results sent by him to Alice). Note that, in order to reach at step $m$, Bob has to pass the test at all the steps $2, ...., m-2, m-1$. Thus, we can bound $P_{BAS}(m)$ from above by the probability to pass Alice's test at step $m-1$ only
\begin{equation}\label{eq:BAS}
\ds \mathcal{P}_{BAS}(m) = \sum_{\ell_A^{(B)}} q(\ell_A^{(B)}) \;\; \times
\underbrace{\sum_{v=\mu - 3\sigma}^{2N - \ell_A^{(B)}} \binom{2N - \ell_A^{(B)}}{v} P_{=}^v  P_{\neq}^{2N - \ell_A^{(B)} - v}}_{\substack{\text{probability for Alice to obtain at least} \\ \text{$\mu - 3\sigma$ correct results on the $2 N - \ell_A^{(B)}$ results} \\ \text{out of the $m-1$ results received from Bob}}},
\end{equation}
with $\mu$ and $\sigma$ defined by~\eqref{eq:meanvar}, where $m_r=2N - \ell_A^{(B)}$. The probability distribution for $\ell_{A}^{(B)}$ given by
\begin{equation}
q\left(\ell_{A}^{(B)}\right)=\frac{\dbinom{m-1}{2N-\ell_{A}^{(B)}} \dbinom{4N-m+1}{\ell_{A}^{(B)}}}{\dbinom{4N}{2N}}.
\end{equation}
$P_{ABS}(m)$, the probability that Alice passes Bob's test on their Shared qubits (from the $m$ results she sends to Bob), as well as its bound $\mathcal{P}_{ABS}(m)$, are defined analogously. Thus, the overall probability to reach at step $m$ is given by ${P}_{ABS}(m)\cdot {P}_{BAS}(m) \leq \mathcal{P}_{ABS}(m)\cdot \mathcal{P}_{BAS}(m)$. 

Hence, the average probability for Bob to cheat, with both Alice and Bob having reached the step $m$ by passing each other's tests, is
\begin{equation}
\bar{P}_{ch}^{\bob} (m) \leq \mathcal{P}_{ABS}(m) \cdot \mathcal{P}_{BAS}(m) \int_{\alpha} p(\alpha) P_{bind}^\bob(m;\alpha) \left(1 - P^\alice_{bind}(m;\alpha)\right)\; \mbox d \alpha .
\end{equation}

\begin{figure}
\centering
\begin{tabular}{cc}
\includegraphics[width=6.5cm]{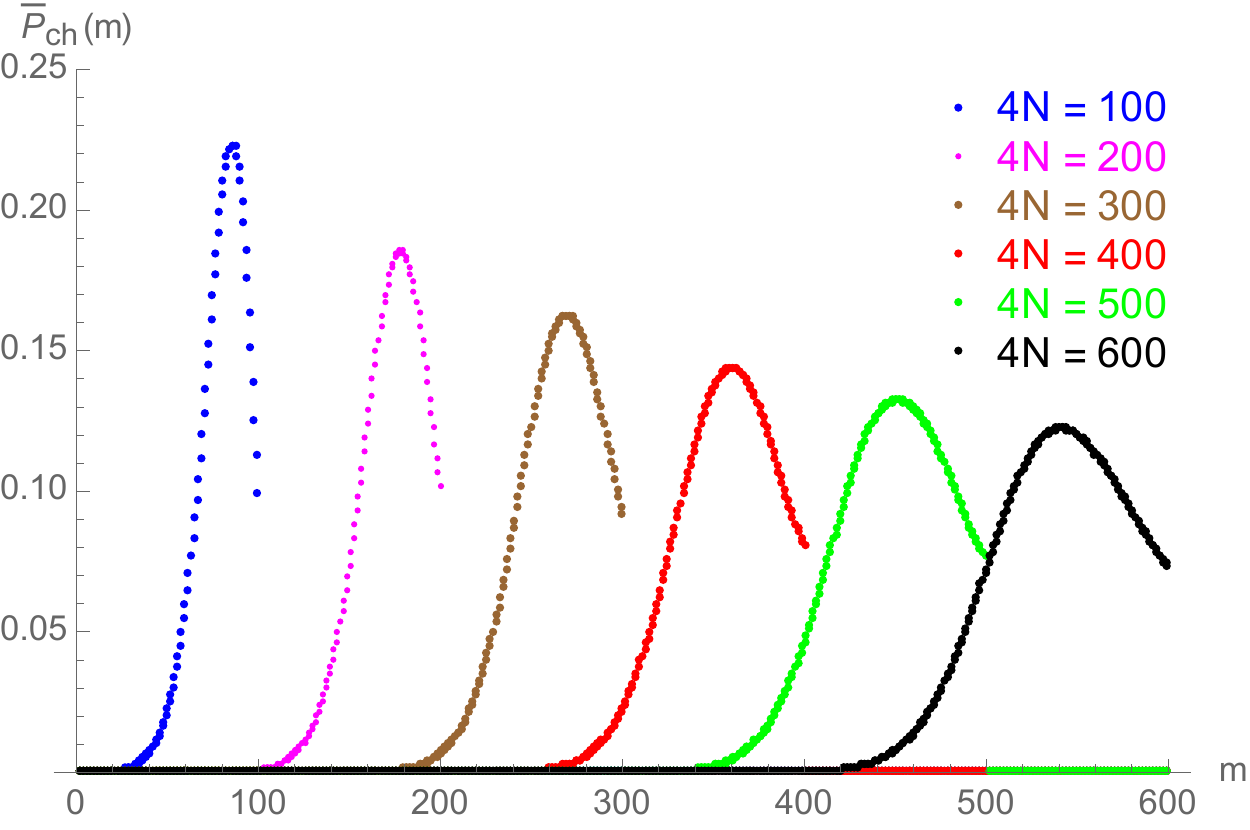}    \ \ \ \     &
\includegraphics[width=6.5cm]{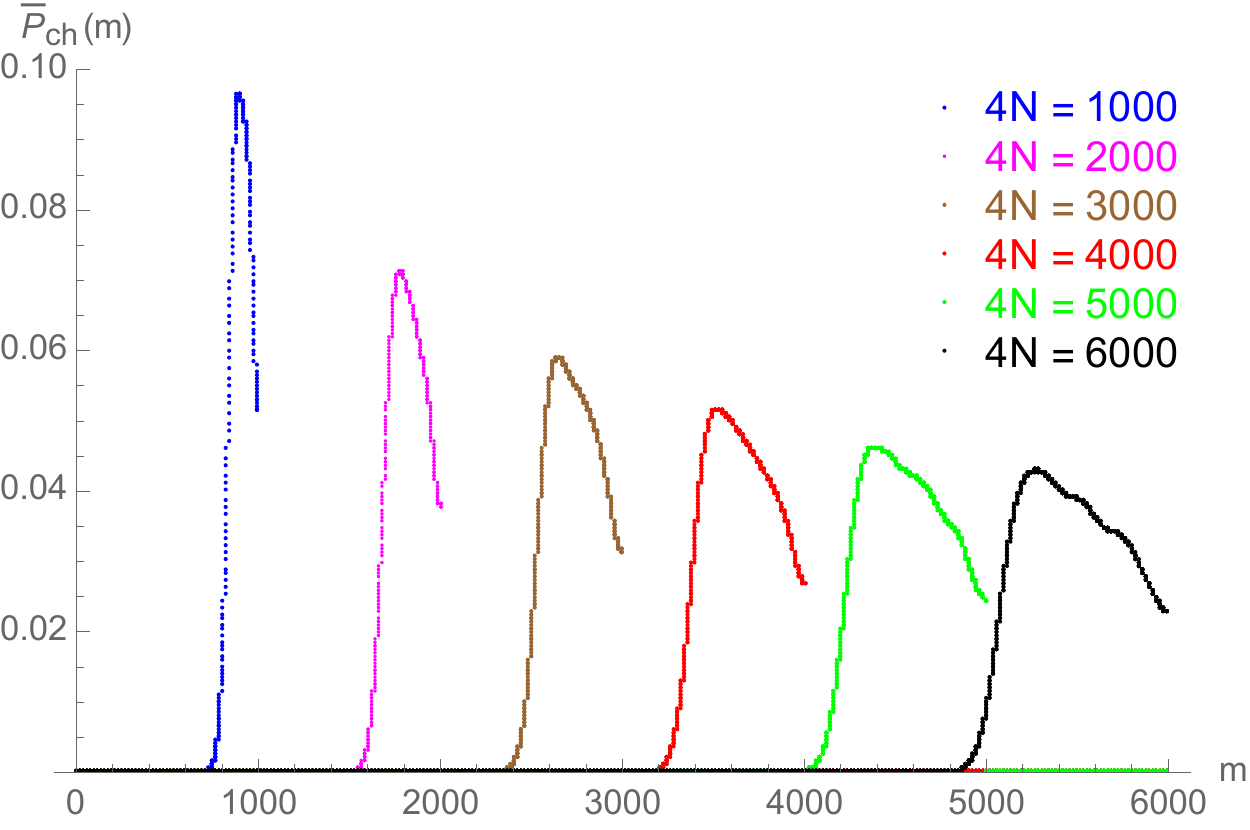} \tabularnewline
\end{tabular}
\caption{\footnotesize{Honest scenario (noisy channel): the expected probability to cheat, $\bar{P}_{ch}(m)$, is plotted against the step of communication interruption $m$ (from $1$ to $4N$), for $\alpha$ chosen uniformly over the interval $[0.9,0.99]$. The noise parameter is $\kappa = 0.05$.}}
\label{fig:honest_noisy2D}
\end{figure}

In Figures~\ref{fig:honest_noisy2D} we present the expected probability to cheat, $\bar{P}_{ch}(m)$, plotted against the total number of communications between Alice and Bob, $4N$, for the 3-sigma acceptance criterion of both Alice and Bob (as before, $p(\alpha)$ is uniform on the interval $[0.9,0.99]$). The maximal value again shows the behavior $\max_m \bar{P}_{ch}(m) \propto N^{-1/2}$, as presented in Figures~\ref{fig:plots}~(b) and~\ref{fig:plots}~(e) from the main text.

\section{Security analysis (probabilistic fairness) against a dishonest client} \label{sec:dishonest_noisy}

In order to cheat, a dishonest client, say Bob, would want to obtain a signed copy of the contract for message $M$ that Alice and Bob initially agreed upon, without letting Alice obtain a signed copy for herself, so that he can use it later on, if he wants to.

The qubits from $\mathcal T^{(A)}_B$ are used to verify Alice's honesty, and therefore Bob measures them according to the protocol (note that he knows which qubits from $\alice$ are entangled with $\mathcal T^{(A)}_B$, and also which observable $H_{\texttt{A}_i}$ to measure). Regarding qubits from $\bob$, unlike the standard quantum cryptographic protocols, where the task of an adversary (say Eve, in key distribution schemes) is to distinguish between mutually non-orthogonal states, in this protocol a cheating Bob knows the pure states of his qubits. He knows that Trent and an honest Alice measure $\hat{H}_{\texttt{B}_i}$ on their respective halves of the pairs entangled with Bob, and thus by measuring $\hat{H}_{\texttt{B}_i}$, he can check which of the two mutually orthogonal states they (will) collapse his system to. Note that Bob cannot take advantage of measuring observables different from those prescribed by the protocol. Since Alice's and Bob's measurements are local, they commute and thus their time order is irrelevant. Thus, by deviating from the protocol, a dishonest Bob can only spoil the correlations between his and Alice's outcomes, a task he can achieve by any random source. Regarding Trent, the time ordering matters, as it is Bob who tells Trent which single-qubit observables to measure, based on the $h^*=h(M)$ he provides Trent with. Nevertheless, since Bob's aim is to have his results as correlated as possible with Trent's, without spoiling the correlations established with Alice's results, Bob should measure his Honest observables $\hat{H}_{\texttt{B}_i}$ on all the qubits from $\bob$.

As argued above, for a cheating client Bob, even knowing to which states the particles sent to him are collapsed to (due to Alice's and Trent's measurements), does not help. This is because he still does not know which of those particles are entangled with Alice and which ones are with Trent. Bob's cheating strategy should allow him to bind the contract, such that Alice is unable to bind the contract even with Trent's help. In order to bind the contract for himself, Bob must pass the test by providing to Trent at least a fraction $\alpha$ of correct measurement results on the qubits entangled with $\mathcal{T}_{T2}^{(A)}$, as well as $\mathcal{T}_{T2}^{(B)}$, respectively. At the same time, Bob does not want Alice to pass the test on the results sent by him (corresponding to qubits entangled with $\mathcal{T}_{T1}^{(B)})$ to her. Making measurements in a basis other than the one given by $\hat{H}_{\texttt{B}_i}$ on some of the qubits and sending those results to Alice does not help him in any way, because then he ends up with incorrect results on some of his qubits in $\mathcal{T}_{T2}^{(B)}$ entangled with Trent. While these incorrect results will spoil Alice's chances to bind the contract, they will equally decrease Bob's chances too. Hence, Bob's best strategy is to measure his Honest observable $\hat{H}_{\texttt{B}_i}$ on all his qubits, to pass the test on all $\mathcal{T}_{T2}^{(B)}$ qubits, and send random bits to Alice, by choosing a random probability $f$ to decide whether or not to flip the result that he sends to Alice (the frequency of sending wrong results). Therefore, the probability for Alice to obtain correct results on these qubits will be
\begin{equation}\label{eq:tilde_P}
\tilde{P}_{=}=(1-f)P_{=} + f P_{\neq},
\end{equation} 
where $P_{=} = p(00) + p(11) $ and $P_{\neq} =p(10) + p(01)$.

\begin{figure}
\centering  
\begin{tabular}{cc}        
\includegraphics[width=6.5cm]{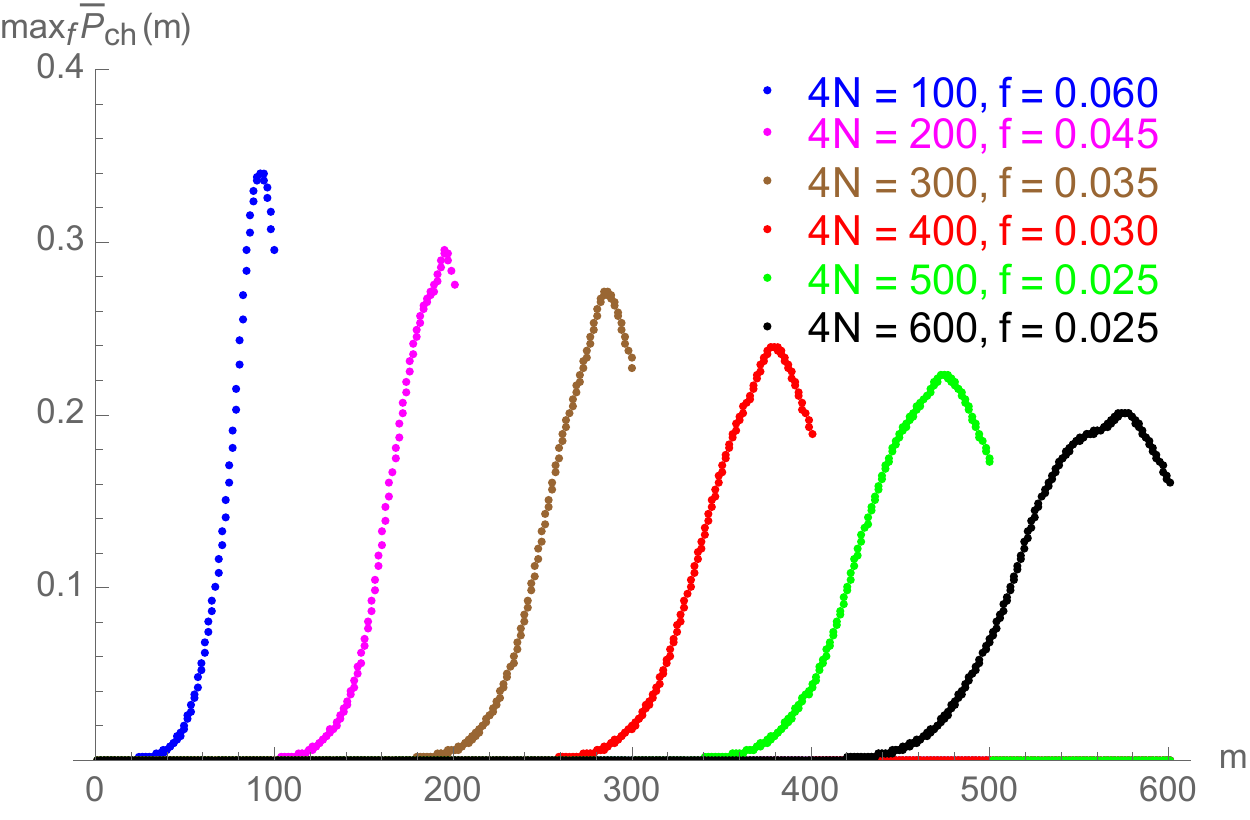}  \ \ \ \      &
\includegraphics[width=6.5cm]{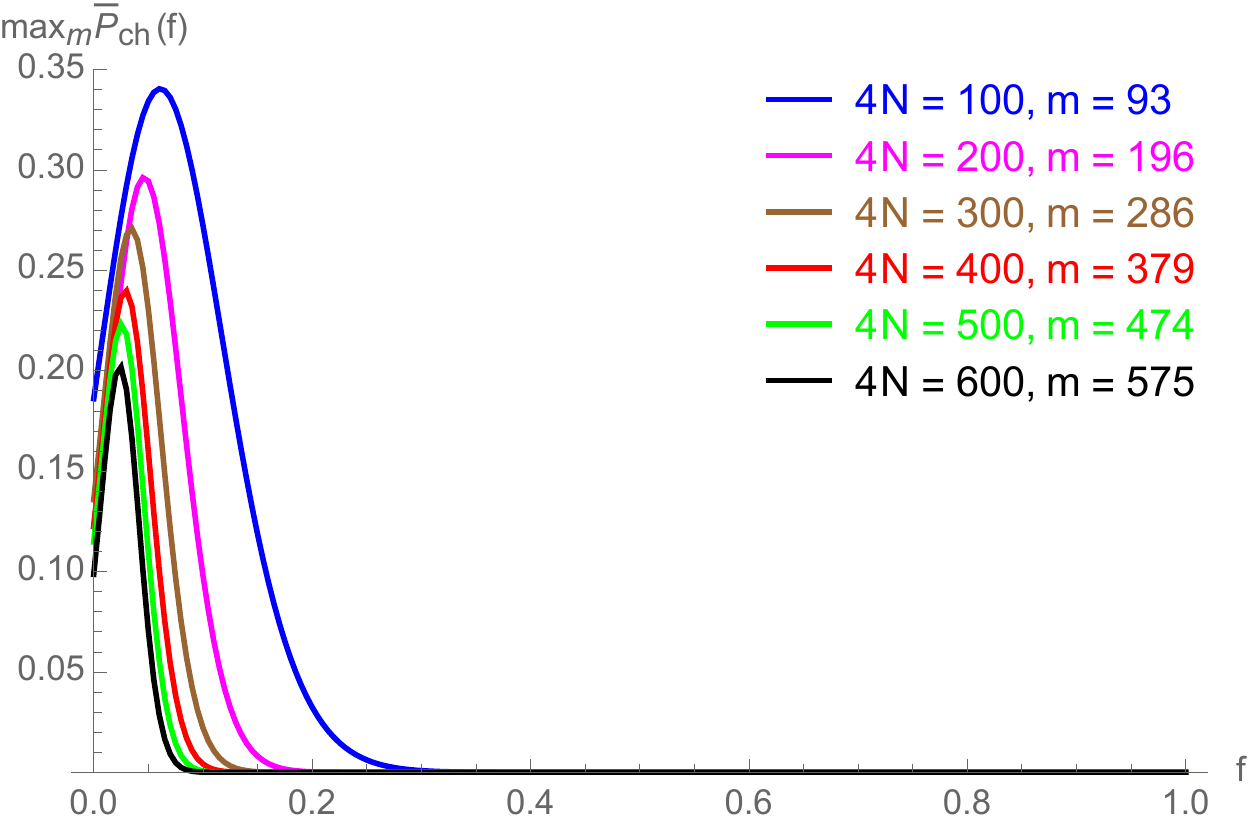}\\
\includegraphics[width=6.5cm]{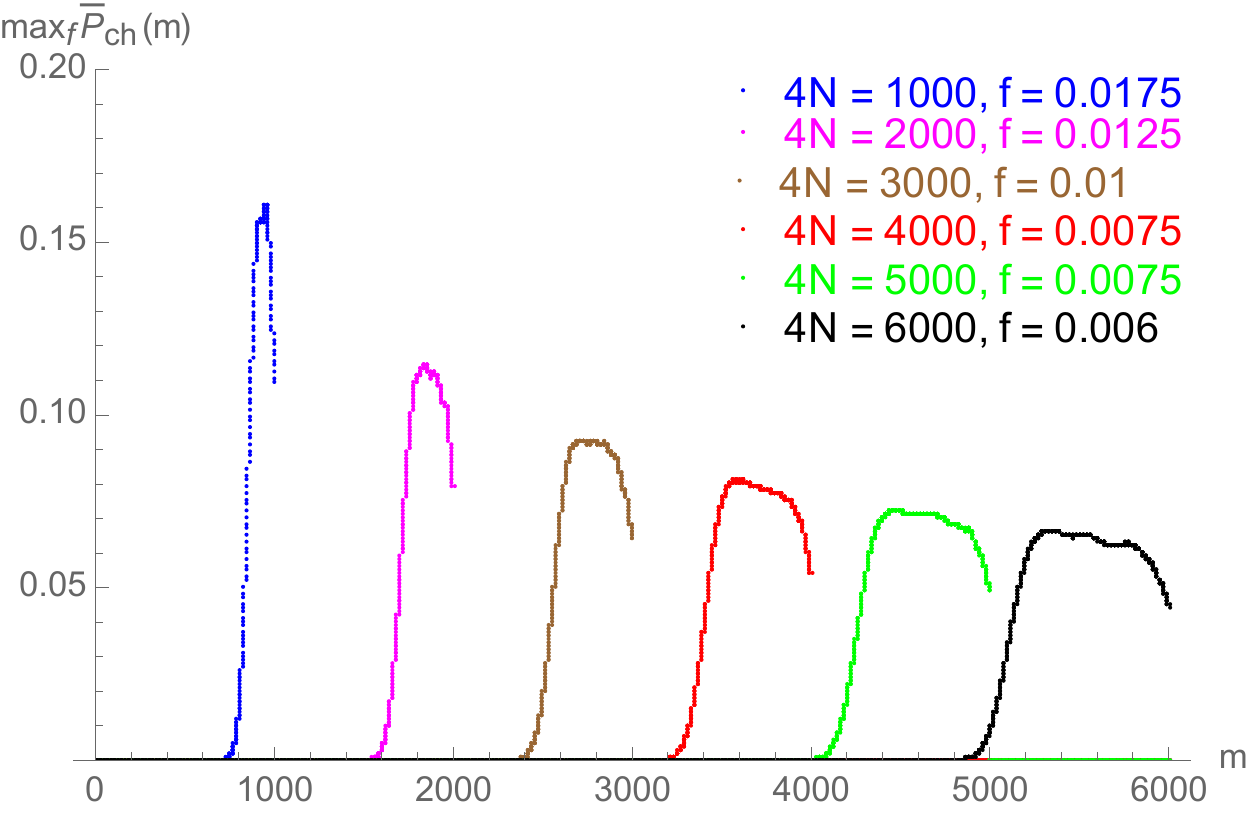}   \ \ \ \     &
\includegraphics[width=6.5cm]{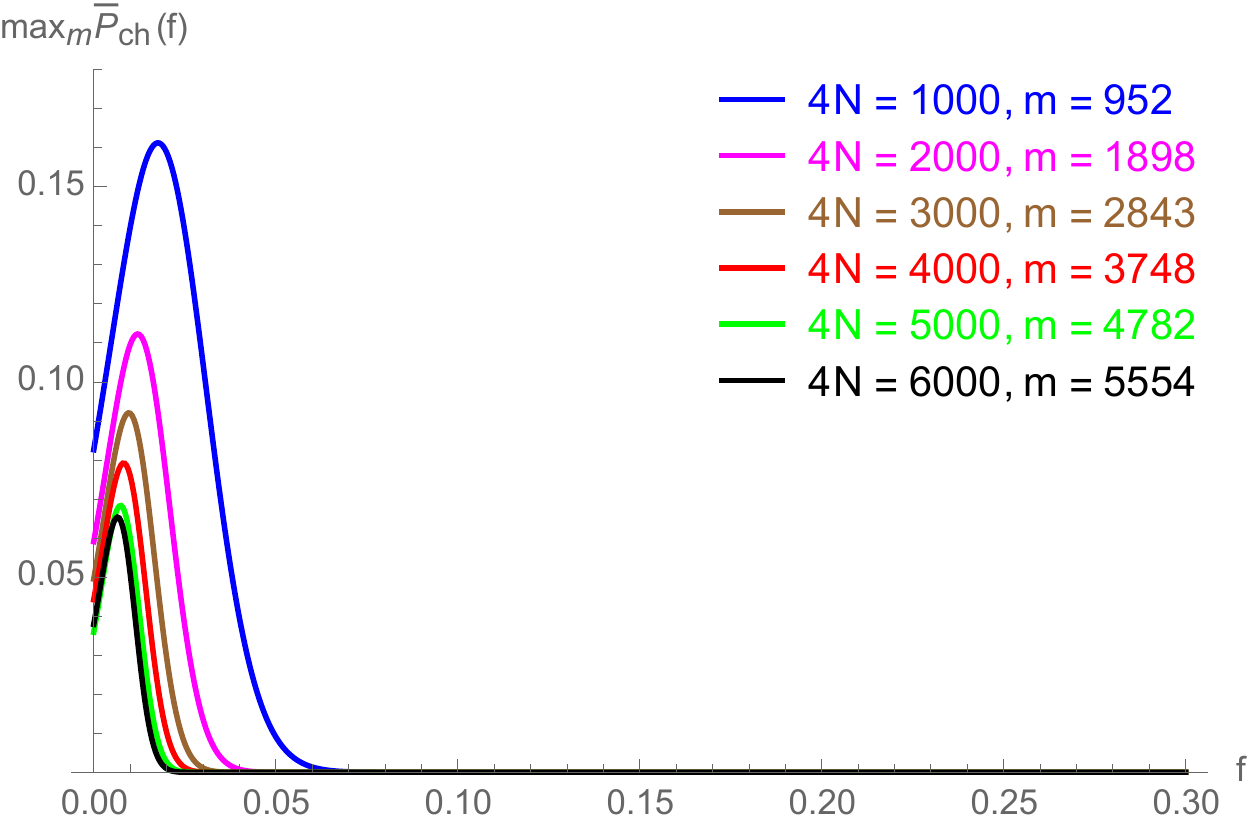}
\end{tabular}
\caption{\footnotesize{(Color online) Dishonest scenario (noisy channel): the maximal expected probability to cheat, $ \bar{P}_{ch}(m,f)$, is first plotted against the step of communication interruption, $m$ (taking the optimal value for $f$), and then against the cheating parameter, $f$ (taking the optimal value for $m$), for the case for the uniform $p(\alpha)$ on the interval $[0.9,0.99]$. The noise parameter is $\kappa = 0.05$.}}
\label{fig:dishonest_noisy2D}
\end{figure}

The communication is interrupted at step $m$, after Alice stops sending the measurement outcomes to Bob upon suspecting a dishonest behavior. Recall that we are considering the case where Alice is the first one to start the communication, and therefore, in the worst case scenario after the $m$-th step Alice and Bob each have $m$ measurement results from each other.

With $m$ measurement results each, the probability for Alice and Bob to pass Trent's test on her/his own qubits, $P_{ATH}(\alpha)$ and $P_{BTH}(\alpha)$ respectively, remains the same as in the honest noisy case. Since Bob decides to flip the results randomly, based on $f$, he is bound to send wrong results on some of the qubits. Hence, Alice receives more incorrect results from Bob as compared to the honest noisy case. Bob on the other hand receives the same number of correct results as in the honest noisy case. Therefore, $\mathcal P_{ABS}(m)$ and $P_{BTA}(m;\alpha)$ remain the same, while $P_{ATB}(m;\alpha)$ and $\mathcal P_{BAS}(m)$ are modified by replacing $P_{=}$ and $P_{\neq}$ by $\tilde{P}_{=}$ and $\tilde{P}_{\neq} = 1 - \tilde{P}_{=}$~\eqref{eq:tilde_P}.

The expected probability to cheat, $\bar{P}_{ch}(m,f)$, is plotted in Figure~\ref{fig:dishonest_noisy2D}. The same behavior (as in the case of honest clients), $\max_{m,f} \bar{P}_{ch}(m,f) \propto N^{-1/2}$, is observed, as in Figures~\ref{fig:plots} (c) and~\ref{fig:plots} (f) from the main text.

\bibliographystyle{unsrt}

\end{document}